\documentclass{article}
\usepackage{amsmath}
\usepackage{authblk}
\usepackage{amssymb, amsthm, amsfonts, tikz-cd, mathrsfs,mathtools,stmaryrd,enumitem, algorithmic,float,comment,tikz}
\usetikzlibrary{backgrounds}
\usetikzlibrary{decorations.pathreplacing}
\usepackage[colorlinks,citecolor=blue,linkcolor=blue,urlcolor=red]{hyperref}
\usepackage{fullpage}
\usepackage{xcolor}
\usepackage[ruled,vlined,linesnumbered]{algorithm2e}
\usepackage{thm-restate}
\usepackage{multirow}
\usepackage{tabularx}
\usepackage{appendix}
\usepackage{placeins}

\usepackage[capitalize]{cleveref}
\crefname{theorem}{Theorem}{Theorems}
\crefname{lemma}{Lemma}{Lemmas}
\crefname{claim}{Claim}{Claims}

\newtheoremstyle{dotless}{}{}{\itshape}{}{\bfseries}{}{ }{}
\theoremstyle{dotless}
\newtheorem{theorem}{Theorem}[section]
\newtheorem{lemma}[theorem]{Lemma}

\newtheorem{claim}[theorem]{Claim}

\newtheoremstyle{dotlessdef}{}{}{\normalfont}{}{\bfseries}{}{ }{}
\theoremstyle{dotlessdef}
\newtheorem{definition}[theorem]{Definition}

 \newcommand{\coord}{\normalfont\texttt{coord}}
 \newcommand{\proj}{\normalfont\texttt{proj}}
  \newcommand{\conv}{\normalfont\texttt{Conv}}

\newcommand{\Oish}{\widetilde{O}}
\newcommand{\wt}{\textup{wt}}

\newcommand{\eps}{\varepsilon}
\newcommand{\dist}{\text{dist}}

\usepackage{todonotes}

\newcommand{\ghnoteinline}[1]{\todo[inline, size=\normalsize, color=blue!20]{Gary's Note: #1}}

\def\ShowAuthNotes{1}
\ifnum\ShowAuthNotes=1
\newcommand{\authnote}[2]{\textcolor{blue}{[{\footnotesize {\bf #1:} { {#2}}}]}}
\else
\newcommand{\authnote}[2]{}
\fi

\title{Additive Spanner Lower Bounds with Optimal Inner Graph Structure}

\author[1]{Greg Bodwin}
\author[1]{Gary Hoppenworth}
\author[2]{Virginia Vassilevska Williams}
\author[1]{Nicole Wein}
\author[2]{Zixuan Xu}
\affil[1]{University of Michigan EECS. \texttt{\{bodwin,garytho,nswein\}@umich.edu}\thanks{Bodwin and Hoppenworth were supported by NSF:AF 2153680.}}
\affil[2]{MIT, EECS. \texttt{\{virgi,zixuanxu\}@mit.edu}\thanks{V. Williams and Xu (partially) were supported by NSF Grant CCF-2330048, BSF Grant 2020356 and a Simons Investigator Award.}}

\date{ }

\begin{document}

\maketitle

% \gbnoteinline{Title brainstorm:

% Additive Spanner Lower Bounds with Optimal Inner Graphs? [add ? to current title]

% Improved Additive Spanner Lower Bounds via Aligned Inner Graphs

% Additive Spanner Lower Bounds with Conditionally Optimal Inner Graphs

% Additive Spanner Lower Bounds with Optimally Sparse Inner Graphs
% }
% \xzx{I think ``aligned inner graphs'' is the only one that is not misleading, but i think it might seem very confusing to the readers on first look.}

% \nnote{I don't love any of these but I don't have better ideas right now. ``Aligned'' sounds like something different than what it actually means, ``conditionally optimal'' is misleading because it sounds like the fine-grained complexity type of condition}

\begin{abstract}
We construct $n$-node graphs on which any $O(n)$-size spanner has additive error at least $+\Omega(n^{3/17})$, improving on the previous best lower bound of $\Omega(n^{1/7})$ [Bodwin-Hoppenworth FOCS '22].
Our construction completes the first two steps of a particular three-step research program, introduced in prior work and overviewed here, aimed at producing tight bounds for the problem by aligning aspects of the upper and lower bound constructions.
More specifically, we develop techniques that enable the use of \emph{inner graphs} in the lower bound framework whose technical properties are provably tight with the corresponding assumptions made in the upper bounds. As an additional application of our techniques, we improve the corresponding lower bound for $O(n)$-size additive emulators to $+\Omega(n^{1/14})$.
\end{abstract}

\section{Introduction}\label{sec:intro}

Suppose that we want to compute shortest paths or distances in an enormous graph $G$.
When $G$ is too big to store in memory, a popular strategy is to instead use a \emph{spanner} of $G$, which is a much sparser subgraph $H$ with approximately the same shortest path metric as $G$.
This can substantially improve storage or runtime costs, in exchange for a small error in the distance information.
Perhaps the most well-applied case is when the spanner is asymptotically as sparse as possible; that is, $|E(H)| = O(n)$ for an $n$-node input graph $G$ (note that $\Omega(n)$ edges are needed just to preserve connectivity).

There are several ways to measure the quality of approximation of a spanner.
The two most popular are as follows:
\begin{definition} [Multiplicative and Additive Spanners]
Given a graph $G$, a subgraph\footnote{Throughout the paper, for brevity, we write ``subgraph'' to specifically mean a subgraph over the same vertex set as the original graph.} $H \subseteq G$ is a \emph{multiplicative $\cdot k$ spanner} if for all nodes $s, t$ we have
$\dist_H(s, t) \le \dist_G(s, t) \cdot k.$
It is an \emph{additive $+k$ spanner} if we have
$\dist_H(s, t) \le \dist_G(s, t) + k.$
\end{definition}

The parameter $k$ is called the (additive or multiplicative) \emph{stretch} of the spanner.
A famous paper of Alth{\" o}fer, Das, Dobkin, Joseph, and Soares \cite{althofer1993sparse} settled the optimal \emph{multiplicative} stretch for $O(n)$-size spanners:\footnote{Although we generally treat input graphs $G$ as undirected and unweighted in this paper, this particular theorem also extends to the setting where $G$ is weighted.}
\begin{theorem} [\cite{althofer1993sparse}]
Every $n$-node graph has a spanner $H$ of size $|E(H)| = O(n)$ and multiplicative stretch $O(\log n)$.
This stretch cannot generally be improved to $o(\log n)$.
\end{theorem}

\begin{table}[t]
\begin{center}
    \begin{tabular}{|l||lc|lc|}
        \hline
        & \multicolumn{2}{c|}{Upper Bound} & \multicolumn{2}{c|}{Lower Bound} \\
        \hline
        \multirow{6}{*}{$O(n)$-size Spanners} & & & $\Omega(\log n)$  & \cite{Woodruff06} \\
                                              & $\Oish(n^{9/16})$   & \cite{Pettie09} & $\Omega(n^{1/22})$ & \cite{AB17jacm}\\
                                              & $\Oish(n^{1/2})$    & \cite{BV15}     & $\Omega(n^{1/11})$ & \cite{HuangP18,Lu19}\\ 
                                              & $O(n^{3/7 + \eps})$ & \cite{BV21}     & $\Omega(n^{2/21})$ & \cite{LuVWX21}\\
                                               & {\bf \color{red} $O(n^{\frac{15-\sqrt{54}}{19} < 0.403})$} & \cite{TZ23} & $\Omega(n^{1/7})$ & \cite{BodwinH22}\\
                                               & & & {\bf \color{red} $\Omega(n^{3/17})$} & this paper\\
        \hline
    \end{tabular}
\caption{\label{tbl:priorwork} The progression of upper and lower bounds on the additive error associated to $n$-node spanners on $O(n)$ edges; current state of the art bounds are highlighted in red.  See also \cite{BCE05, Chechik13, BKMP10, ACIM99, ABP18} for work on additive spanners of superlinear size.}
\end{center}
\end{table}

The goal of this paper is to make progress on the corresponding question for \emph{additive} error.
This question has been intensively studied; see Table \ref{tbl:priorwork} for the progression of results.
Our contributions are on the lower bounds side:

\begin{theorem} [Main Result] \label{thm:spanner} There exists an infinite family of $n$-vertex undirected graphs for which any additive spanner on $O(n)$ edges has additive stretch $\Omega(n^{3/17})$.
\end{theorem}

Our techniques also lead to progress on related questions for $O(n)$-size emulators, which we discuss further in Section \ref{sec:addresults}.
Before we explain this, we contextualize Theorem \ref{thm:spanner} by explaining in more depth the sense in which it moves the upper and lower bounds closer together.

\subsection{Our Contribution and Next Steps for the Area}
\label{subsec:program}

There are well-established frameworks in place for proving upper and lower bounds for $O(n)$-size spanners, and the current sentiment among experts is that these two frameworks \emph{could} eventually produce near-matching (likely within $n^{\eps}$ factors) upper and lower bounds.
Both frameworks can be broken down into three corresponding steps, and over the last few years, a research program has emerged in which the long-term goal is to find optimal bounds for the problem by making each of these three steps align.\footnote{This program was made somewhat explicit in \cite{BodwinH22} (c.f.\ Section 2.4), but was implicit in work before that.}
That is, we can investigate what ``should'' happen in each step if a hypothetical optimal version of the upper bound framework were run on the graph from a hypothetical optimal version of the lower bound framework.
This thought experiment leads to a list of three concrete features that should be realized in an ideal lower bound, which we overview at a high level in Table \ref{tbl:alignment}.

However, it is easier to write down this wishlist for the lower bound than it is to actually achieve the listed features in a construction; we discuss the various technical barriers in Section \ref{sec:overview}.
The contribution of the current paper is to achieve the first two steps of alignment (i.e., the first two items in Table \ref{tbl:alignment}) simultaneously, which both have to do with optimizing properties of the so-called \emph{inner graph} in the lower bound construction.
That said, the ideal structure of an inner graph has been known since \cite{BV21}, and well before that Coppersmith and Elkin \cite{CoppersmithE06} found graph constructions achieving this ideal structure (``subset distance preserver lower bound graphs'').
Our main technical contributions are not in designing new inner graphs, but rather, in improving the \emph{outer graph} in a way that allows these previously known optimal inner graphs to be used within the framework for the first time.

This paper makes no real progress on the third and final point of alignment, which contends with optimizing certain quantitative properties of the shortest paths in the outer graph. 
Here there is still significant misalignment between the upper and lower bounds, which is responsible for essentially all of the remaining numeric gap between the current upper and lower bounds for $O(n)$-size spanners.
Improving this third point, either on the upper bounds side or the lower bounds side, is the clear next step for the area and it may first require advances in our understanding of distance preservers \cite{CoppersmithE06}; see \cite{BodwinH22} for discussion.

\begin{table}[t]
\begin{center}
\begin{tabularx}{\textwidth}{X|X|X}
Step in Upper Bounds & Step in Lower Bounds & What should ideally happen when we run the upper bound framework on a lower bound graph?\\
\hline
& & \\
Cover the input graph by clusters $C$ of radius $r$ each.
These clusters are classified as either \emph{small} or \emph{large}, depending on whether their number of nodes is smaller or larger than $r^{4/3}$. & Start with an \emph{outer graph}, and systematically replace each node with a disjoint copy of an \emph{inner graph}. & The upper bound should select the inner graphs as its clusters.  All inner graphs should have $\Theta(r^{4/3})$ nodes, since the worst case for the upper bound is when all clusters are near the large/small threshold.\\
& & \\
Small clusters $C$ have a node separator of size $\le |C|^{1/4}$.
Construct a \emph{subset distance preserver} on each small cluster, preserving all shortest paths between separator nodes, at cost $O(|C|)$ \cite{CoppersmithE06}. & The inner graphs should be selected as the union of many long unique shortest paths among nodes that form a separator for the graph, and also any two of these shortest paths may intersect on at most one node. & The inner graph should be a lower bound graph against subset distance preservers with $\Theta(|C|^{1/4})$ source nodes (with a large implicit constant), so that the approach of constructing a subset distance preserver is too expensive to be used in an attack against the lower bound.\\
& & \\
Large clusters $C$ are handled by adding some additional shortest paths in the spanner to connect far-away clusters to each other.
Using the \emph{path-buying framework} \cite{BKMP10}, we can limit the total number and length of the shortest paths we need to add. & The outer graph is selected to be the union of as many long unique shortest paths as possible, and any two of these shortest paths may intersect on at most one edge.  That is, the outer graph is a slightly modified distance preserver lower bound graph. & The shortest paths added for large clusters should coincide with the shortest paths in the original outer graph (before inner graph replacement).
The path-buying bounds on the number and length of these shortest paths should coincide with the number and length of these shortest paths in the outer graph.
\end{tabularx}
\end{center}
\caption{\label{tbl:alignment} A point-by-point comparison of the frameworks used to prove upper and lower bounds.  Our main technical contributions are to satisfy the first point of alignment by enabling the use of inner graphs with $\Theta(r^{4/3})$ nodes (where $+\Omega(r)$ is the desired lower bound on spanner error), and to satisfy the second point of alignment by enabling the use of subset distance preserver lower bounds for our inner graphs.  Neither of these properties were fully achieved in prior work.
}
\end{table}

\FloatBarrier

\subsection{Additional Results \label{sec:addresults}}

The technical improvements to the construction that enable our improved spanner lower bounds also imply improvements for two nearby objects, which we overview next.
First, an \emph{emulator} is similar to a spanner, but not required to be a subgraph:

\begin{definition} [Additive Emulators]
Given a graph $G$, a graph $H$ on the same vertex set as $G$ is an \emph{additive $+ k$ emulator} if for all nodes $s, t$ we have
$$\dist_G(s, t) \le \dist_H(s, t) \le \dist_G(s, t) + k.$$
\end{definition}
An emulator $H$ is allowed to be weighted, even when the input graph $G$ is unweighted.
Emulators generalize spanners, and hence the upper and lower bounds known for $O(n)$-size emulators are a bit lower than the corresponding bounds for spanners.
See Table \ref{tbl:priorworkem} for the progression of results on the additive error that can be obtained for $O(n)$-size emulators.

\begin{table}[h]
\begin{center}
    \begin{tabular}{|l||lc|lc|}
        \hline
        & \multicolumn{2}{c|}{Upper Bound} & \multicolumn{2}{c|}{Lower Bound} \\
        \hline
        \multirow{5}{*}{$O(n)$-size Emulators} & $O(n^{1/3 + \eps})$  & \cite{BV15}     & $\Omega(\log n)$ & \cite{Woodruff06} \\
                                               & $O(n^{3/11 + \eps})$ & \cite{BV21}     & $\Omega(n^{1/22})$ & \cite{AB17jacm}\\
                                               & $\Oish(n^{1/4})$     & \cite{Pettie09} & $\Omega(n^{1/18})$ & \cite{HuangP18}\\
                                               & $\Oish(n^{2/9 - 1/1600 < 0.222})$ & \cite{KP23} & $\Omega(n^{2/29})$ & \cite{LuVWX21}\\
                                               & {\color{red} \bf $O(n^{\frac{1}{3+\sqrt{5}}+\eps<0.191})$} & \cite{Hoppenworth24} & {\bf \color{red} $\Omega(n^{1/14})$} & this paper\\
        \hline
    \end{tabular}
\caption{\label{tbl:priorworkem} The progression of upper and lower bounds on the additive error associated to $n$-node emulators on $O(n)$ edges; current state of the art bounds are highlighted in red.  See also \cite{DHZ00}.}
\end{center}
\end{table}

A similar lower bound framework is used to achieve lower bounds for emulators, and hence our new technical machinery improves the current lower bounds for emulators as well:
\begin{theorem} \label{thm:emulator}
There exists an infinite family of $n$-vertex undirected graphs for which any additive emulator on $O(n)$ edges has additive stretch $\Omega(n^{1/14})$.
\end{theorem}

Our numeric improvement in the lower bound for emulators is more modest than our improvement for spanners; at a high level, this is because our main improvement is to enable stronger inner graphs in the lower bound framework, but the role of the inner graph is generally less important in emulator lower bounds.

We next provide a more fine-grained overview of our lower bound framework, and we describe our technical improvements that lead to our new results in more detail.
\section{Technical Overview}\label{sec:overview}

In this section we will give an overview of the different technical components in our lower bound graph construction. We start by reviewing the obstacle product framework in \cref{subsec:obs-product} and recalling some ideas from prior work in \cref{subsec:overview-outer}. Finally we will discuss the new components in our construction in \cref{subsec:overview-our-construction}.

\subsection{The obstacle product framework}\label{subsec:obs-product}

Similar to all previous works including \cite{AB17jacm,HuangP18,LuVWX21,BodwinH22} on proving stretch lower bounds for linear-sized additive spanners, our construction falls under the \emph{obstacle product} framework introduced in \cite{AB17jacm}. Any construction under this framework consists of an \emph{outer graph} $G_O = (V_O, E_O)$ and an \emph{inner graph} $G_I = (V_I, E_I)$ where every vertex in the outer graph is replaced by a copy of the inner graph. The desired outer graph should contain a set pairs $P_O\subseteq V_O\times V_O$ often called the \emph{critical pairs} such that the following holds:
\begin{enumerate}
    \item For each pair $(s,t)\in P_O$, the shortest path from $s$ to $t$ is \emph{unique}. These unique shortest paths connecting between pairs in $P_O$ are often called the \emph{critical paths}.
    \item The critical paths have roughly the same length $\Theta(k)$.
    \item The critical paths are pairwise edge disjoint.
\end{enumerate}
When we replace each vertex in the outer graph with a copy of the inner graph $G_I$, we make sure that the critical paths remain the unique shortest paths between their endpoints and pairwise edge-disjoint by attaching each incoming edge and outgoing edge to distinct vertices of the inner graph. Finally, we subdivide the edges originally in $G_O$ into paths of length $\Theta(k)$. Now in the resulting graph denoted as $G_{obs} = (V_{obs}, E_{obs})$ with critical pairs $P_{obs}$, each critical path between the endpoints in $P_{obs}$ uniquely corresponds to a critical path in $G_O$ and it takes the form of traveling alternatingly between subdivided edges in $G_O$ and paths in $G_I$. In particular, each critical path travels through $\Theta(k)$ subdivided paths of length $\Theta(k)$, and $\Theta(k)$ inner graph copies.

Now let us see how to show that any sparse spanner on $G_{obs}$ must suffer additive distortion $+\Omega(k)$. The goal is to argue that if lots of edges are missing in the spanner $H\subseteq G_{obs}$ compared to $G_{obs}$, then there exists some pair $(s,t)\in P_{obs}$ whose shortest path $\pi$ in $H$ falls into one of the following two cases:
\begin{enumerate}
    \item If $\pi$ traverses the same sequence of inner graph copies as the critical path in $G_{obs}$, then it must use at least one {\em extra} edge in each inner graph copy compared to the critical path in the original graph due to missing edges. Since the critical path passes through $\Theta(k)$ inner graph copies, the path $\pi$ must suffer a $+\Omega(k)$ distortion in total.

    \item If $\pi$ traverses a different sequence of inner graph copies, then it must traverse a different set of subdivided paths corresponding to the edges in $G_O$. Since the critical paths in $G_O$ are the unique shortest paths between its endpoints, $\pi$ must traverse at least one more subdivided path of length $\Theta(k)$ and thus suffer a $+\Omega(k)$ distortion.
\end{enumerate}

Furthermore, note that the reason behind doing inner graph replacement is that without the inner graphs, subdividing each edge of the outer graph would significantly sparsify the graph so that even a trivial spanner including all the edges would have linear size. Adding the inner graphs helps balance the overall density of the graph so that any linear-sized spanner needs to be nontrivial. Thus, ideally we would want the inner graphs to be dense.

\subsection{The outer graph: distance preservers and the alternation product}\label{subsec:overview-outer}

In this subsection, we review the two key components for the outer graph construction: the distance preserver lower bound graph given in \cite{CoppersmithE06} and the alternation product first used in \cite{Hesse03}.

\paragraph{Distance preserver lower bound graph.}

 Given a graph $G = (V,E)$ and a set of pairs $P\subseteq V\times V$, a distance preserver $H$ is a sparse subgraph of $G$ that preserves the distances for every pair in $P$ \emph{exactly}. Previously, Coppersmith and Elkin \cite{CoppersmithE06} obtained a lower bound instance for distance preservers by constructing a large set of vertex pairs with pairwise {\em edge-disjoint} unique shortest paths that are as long as possible; the union of the edges of these paths is the lower bound instance.
Following the intuition outlined in \cref{subsec:obs-product}, it is natural to consider using the distance preserver lower bound construction of \cite{CoppersmithE06} as the outer graph. 
From now on, we will abbreviate the term ``distance preserver lower bound graph''  to ``DP LB graph'' and the term ``Coppersmith-Elkin construction'' to ``CE construction'' for convenience.

Indeed, all prior work uses some version of the CE construction of DP LB graph as the outer graph, and so do we.
The CE construction is a geometric construction where the vertex set corresponds to a $d$-dimensional integer grid $[n]^d$ and edges are added corresponding to a $d$-dimensional convex set $B_d(r)$ defined to be the vertices of the convex hull of integer points contained in a ball of radius $r > 0$. More specifically, the vertices corresponding to the points $\vec{x}, \vec{y}$ are connected by an edge if $\vec{y} - \vec{x}\in B_d(r)$. Then the critical paths are defined to be the paths corresponding to the straight lines starting from a ``start zone'' passing through the grid, i.e. the paths that repeatedly take the edge corresponding to the same vector in $B_d(r)$. By convexity of the set $B_d(r)$, one can show that these critical paths are edge-disjoint and they are the unique shortest paths between their endpoints.

Prior to the work of \cite{BodwinH22}, works including \cite{HuangP18,LuVWX21} all considered a {\em layered} version of the DP LB graph as the outer graph. Namely the graph contains $\ell+1$ layers where each layer corresponds to a $d$-dimensional integer grid $[n]^d$ and edges are added between adjacent layers corresponding to the convex set $B_d(r)$ similarly as defined in \cite{CoppersmithE06}. Then the critical paths are defined to be the paths that start in the first layer and end in the last layer that repeatedly take the edge corresponding to the same vector in $B_d(r)$. This layering simplifies the stretch analysis for additive spanner lower bounds because it is easy to argue that all the critical paths have length exactly $\ell$ and the shortest path should not take any backward edges as it will then need to traverse more layers. However, the layered version resulted in worse bounds compared to the original unlayered version but it was unclear at the time how to analyze an unlayered outer graph. Most recently, Bodwin and Hoppenworth \cite{BodwinH22} developed a new analysis framework and successfully analyzed an obstacle product graph with a modified version of the unlayered DP LB graph as the outer graph. As a result, they improved the lower bound to $\Omega(n^{1/7})$ from $\Omega(n^{1/10.5})$ where the former remains the current best known lower bound before this work. We use the unlayered outer graph construction in \cite{BodwinH22} as an ingredient in our construction.

\paragraph{The alternation product.}

Another important idea that goes in to the outer graph construction is the alternation product first used in \cite{Hesse03}. Subsequent works including \cite{AB17jacm, HuangP18, LuVWX21} all use the alternation product in the outer graph construction. Consider two copies $G_1, G_2$ of the same $2$-dimensional layered DP LB graph with $\ell+1$ layers and convex set $B_2(r)$. Namely, each layer corresponds to the $[n]^2$ grid and the edges correspond to the $2$-dimensional convex set $B_2(r)$ of radius $r$. 
The original implementation of the alternation product graph $G_{alt}$ used in \cite{Hesse03, AB17jacm, HuangP18} of $G_1$ and $G_2$ is a graph on $2\ell+1$ layers with each layer corresponding to the $4$-dimensional grid $[n]^4$. Each vertex in $G_{alt}$ corresponds the pair $(v_1,v_2)$ where $v_1\in G_1, v_2\in G_2$ and the edges are added alternatingly between adjacent layers according to $G_1$ and $G_2$, respectively. Specifically, between layer $i$ and $i+1$ for $i$ odd, we connect the vertex $(\vec{x}, \vec{y})$ for $\vec{x}, \vec{y}\in [n]^2$ to  $(\vec{x}+\vec{w}, \vec{y})$ for $\vec{w}\in B_2(r)$; for $i$ even, we connect the vertex $(\vec{x}, \vec{y})$ to $(\vec{x}, \vec{y}+\vec{w})$ for $\vec{w}\in B_2(r)$.  In other words, $G_{alt}$ keeps track of $G_1$ using the first two coordinates and $G_2$ using the last two coordinates. Then a critical path $\pi$ in $G_{alt}$ corresponds to a pair of critical paths $\pi_1$ in $G_1$ and $\pi_2$ in $G_2$ by taking alternating steps from $\pi_1$ and $\pi_2$. So the main advantage of the alternation product for us is that it gives an extra product structure over the set of critical paths that we want in our construction.

Unlike in our construction, prior works including \cite{Hesse03,AB17jacm,HuangP18,LuVWX21} apply the alternation product in order to obtain a different relative count between the number of vertices and the number of critical pairs rather than to obtain the extra product structure.
However, these changes in parameters are in fact unfavorable to the construction for linear-sized spanner lower bounds. To see this, notice that one can equivalently think of $G_{alt}$ as $4$-dimensional CE construction graph using the smaller convex set $\{(\vec{w}_1, \vec{w}_2)\mid \vec{w}_1, \vec{w}_2\in B_2(r)\}$ instead of $B_4(r)$, which means that $G_{alt}$ has fewer critical pairs (see \cref{subsec:overview-our-construction} for a more detailed discussion). In fact, in \cite{HuangP18}, Huang and Pettie gave an $\Omega(n^{1/11})$ lower bound construction without the alternation product that improved on their own construction that uses the alternation product which gave a bound of $\Omega(n^{1/13})$ in the same paper. Later in \cite{LuVWX21}, Lu, Vassilevska Wiliams, Wein and Xu improved on the alternation product that reduces the loss in the number of critical pairs compared to the CE construction, thereby obtaining an $\Omega(n^{1/10.5})$ lower bound that improved on the previous best bound of $\Omega(n^{1/11})$. Most recently, Vassilevska Wiliams, Xu and Xu implicitly constructed an alternation product graph in their $O(m)$-shortcut lower bound construction in \cite{VWXX24} that asymptotically matches the number critical pairs in the CE construction. Unfortunately, their construction is under a different setting so we cannot directly apply their technique to our construction as a blackbox. However, by isolating a main observation implied in their work, we were able to integrate such an alternation product into our construction (see \cref{subsec:overview-our-construction}).

\subsection{Our construction: optimal unlayered alternation product and optimal inner graph structure}\label{subsec:overview-our-construction} 

Our main technical contribution is a linear-sized additive spanner lower bound construction that carefully combines the following ideas: 
\begin{enumerate}
    \item An unlayered DP LB graph as the outer graph, as in \cite{BodwinH22}.
    \item An optimal alternation product implicit in \cite{VWXX24}.
    \item An optimal subset DP LB graph as the inner graphs, as motivated in \Cref{tbl:alignment}.% as suggested by the upper bound algorithm.
\end{enumerate}
We start with comparing our construction with the previously known lower bound constructions in \cref{tab:tech}.

\begin{table}
\centering
{\def\arraystretch{1.2}
\begin{tabular}{|p{4cm}|c|c|c|}
\hline 
   Citation  & Lower bound & Outer graph & Inner graph \\ \hline
     Woodruff \cite{Woodruff06} & $\Omega(\log n)$ &  Butterfly & Biclique \\ \hline
     Abboud, Bodwin \cite{AB17jacm} & $\Omega(n^{1/22})$ & Layered DP LB + Alt Product & Biclique \\ \hline
     Huang, Pettie \cite{HuangP18} & $\Omega(n^{1/13})$ & Layered DP LB + Alt Product & Biclique \\ \hline
     Huang, Pettie \cite{HuangP18} & $\Omega(n^{1/11})$ & Layered DP LB & Layered DP LB \\ \hline
     Lu, Vassilevska W., Wein, Xu \cite{LuVWX21} & $\Omega(n^{1/10.5})$ & Layered DP LB $+$ Improved Alt Product & Biclique \\ \hline
     Bodwin, Hoppenworth \cite{BodwinH22} &  $\Omega(n^{1/7})$ & Unlayered DP LB & DP LB \\ \hline
     \textbf{This work} & $\Omega(n^{3/17})$ & Unlayered DP LB + Optimal Alt Product & Subset DP LB \\ \hline
\end{tabular}}
\caption{All known lower bound constructions}
\label{tab:tech}
\end{table}

In the following, we will discuss the main components of our construction.

\paragraph{Outer Graph: Unlayered DP LB graph with optimal alternation product.}

As mentioned in \cref{subsec:overview-outer}, we would like to be able to apply the implicit alternation product in \cite{VWXX24} to unlayered DP LB graphs. By isolating the main idea that one can use the set $\{(x,y,x^2+y^2)\mid x,y\in [r]\}$ as the convex set in the alternation product graph, we are able to apply the implicit alternation product in \cite{VWXX24} on unlayered DP LB graphs successfully after certain modifications (See \cref{subsec:W} for more details). In the following, we give a more detailed discussion of the informal intuition behind why the alternation product we use is more desirable than the alternation product used in prior works including \cite{AB17jacm,HuangP18,LuVWX21}.

Recall that in \cref{subsec:overview-outer}, one may view an alternation product graph as a CE construction with a different convex set that determines the set of edges of the graph. In addition, a vector from the convex set and a vertex in the ``start zone'' determines a critical path, so we get more critical pairs if we use a larger convex set. More precisely, we want to use a convex set that is ``large'' with respect to the total number of integer points contained in the convex hull of the set. We recall from \cite{BaranyL98} that $|B_d(r)| = \Theta(r^{d\cdot \frac{d-1}{d+1}})$. Let us compare the convex sets used in the various alternation product graphs against $B_d(r)$ in the same number of dimension that is scaled to contain the same number of points in its convex hull asymptotically in \cref{tab:alt-prod}. Then we can see from \cref{tab:alt-prod} that all prior constructions use a convex set that contains less points than the respective $B_d(r)$ while the convex set we use in this work matches the quality of $B_3(r)$, which is optimal in $3$-dimensions (see \cite{BaranyL98} for more details). That is, the construction that we use is as good as the CE construction in $3$-dimensions. 

One may wonder why we do not simply use the CE construction as our outer graph. The reason is that the alternation product has extra structure that is crucial for allowing us to use our desired inner graph. The CE construction lacks these properties. We elaborate on this below.

\begin{table}[]
\centering
{\def\arraystretch{1}
\begin{tabular}{|p{3cm}|c|c|c|}
\hline\hline
Citation  & Convex Set Used & Size of Convex Set & Size of Convex Hull  \\ \hline
4-dim \cite{CoppersmithE06}  & $B_4(r)$ & $\Theta(r^{12/5})$ & $\Theta(r^4)$  \\ \hline
\cite{AB17jacm, HuangP18}  & $\{(\vec{x},\vec{y})\mid \vec{x}, \vec{y}\in B_2(r)\}$ & $\Theta(r^{4/3})$ & $\Theta(r^4)$ \\ \hline\hline

3-dim \cite{CoppersmithE06}  & $B_3(r)$ & $\Theta(r^{3/2})$ & $\Theta(r^3)$  \\ \hline 
\cite{LuVWX21} & $\{(x_1,x_2+y_1,y_2) \mid \vec{x},\vec{y}\in B_2(r)\}$ & $\Theta(r^{4/3})$ & $\Theta(r^3)$ \\ \hline\hline

3-dim \cite{CoppersmithE06}  & $B_3(r^{4/3})$ & $\Theta(r^{2})$ & $\Theta(r^4)$ \\ \hline
\textbf{This work} (based on \cite{VWXX24}) & $\{(x,y,x^2+y^2)\mid x,y\in [r]\}$ & $\Theta(r^2)$ & $\Theta(r^4)$  \\ \hline

\end{tabular}
}
\caption{Comparison between known constructions of the alternation product and the CE construction. The top row in each pair is the corresponding CE construction in the same number of dimensions and scaled to contain the same number of points in its convex hull. The bottom row in each pair indicates the alternation product construction used in the work cited.}
\label{tab:alt-prod}
\end{table}

\begin{center}

\end{center}

\paragraph{Inner graph: Optimal subset DP LB graph.}

For our inner graphs, we use the CE construction of subset DP LB graphs in the regime where the pairs $S\times S$ has size $|S| = \Theta(n^{1/4})$ where $n$ denotes the number of vertices in the graph. In fact, this construction is tight in the sense that it has $\Omega(n)$ edges while on the other hand it is known that there exists subset distance preservers of size $O(n)$ for every set of sources $S$ of size $O(n^{1/4})$. So not only are we using an inner graph structure that aligns with the upper bound algorithm as illustrated in \cref{tbl:alignment}, we are in fact using a tight construction of the desired structure. 

The main reason why we are able to use subset DP LB graphs as inner graphs in our construction is that we have an alternation product graph as our outer graph. We discuss below why an alternation product is necessary for using subset DP LB graphs as inner graphs. In the inner graph replacement step under the obstacle product framework, we need to attach the incoming edges and outgoing edges adjacent to a vertex $v$ to vertices in the corresponding inner graph copy so that each critical path passing through $v$ in the outer graph will pass through a unique critical path in the inner graph copy as well. Since the subset DP LB graph has critical pairs of the form $S\times S$ for some subset $S$ of the vertex set, it is required that the critical paths passing through $v$ in the outer graph also be equipped with a product structure. In DP LB graphs, we have no such product structure over the critical paths. However, notice that applying an alternation product would exactly give us a product structure over the critical paths as desired.
\section{Preliminaries}

We use the following notations:
\begin{itemize}
    \item We use $\conv(\cdot)$ to denote the convex hull of a set. 
    \item We use $\langle \cdot, \cdot \rangle$ to denote the standard Euclidean inner product, $\| \cdot \|$ the Euclidean norm, and  $\proj_{\Vec{w}}(\cdot)$  the Euclidean scalar projection onto $\Vec{w}$.
    \item We use $[x, y],$ where $x \leq y \in \mathbb{Z}$,  to denote the set $\{x, x+1, \dots, y-1, y\}$. We use $[x]$, where $x > 0 \in \mathbb{Z}$, to denote $[1, x]$. 
\end{itemize}

\section{Outer Graph $G_O$}

\label{sec:outer_graph}

The goal of this section will be to construct the outer graph $G_O$ of our additive spanner and emulator lower bound constructions. The key properties of $G_O$ are summarized in Theorem \ref{thm:outer_graph}.

\begin{theorem}[Properties of Outer Graph]
\label{thm:outer_graph}
For any $a, r > 0 \in \mathbb{Z}$, there exists a graph $G_O(a, r) = (V_O, E_O)$ with a set $\Pi_O$ of critical paths in $G_O$ that has the following properties:
\begin{enumerate}
    \item The number of nodes, edges, and critical paths in $G_O$ is:
    \begin{align*}
        |V_O| &=  \Theta(a^3 r), \\
        |E_O| &=   \Theta(a^3r^2),\\
        |\Pi_O| &=  \Theta(a^2r^4).
    \end{align*}
    \item Every critical path $\pi \in \Pi_O$ is a unique shortest path in $G_O$ of length at least $|\pi| \geq \frac{a}{4r}$.
    \item Every pair of distinct critical paths $\pi_1, \pi_2 \in \Pi_O$ intersect on at most two  nodes. 
    \item Every edge $e \in E_O$ lies on some critical path in $\Pi_O$. 
\end{enumerate}
\end{theorem}

The rest of the section is devoted to constructing the graph $G_O(a,r)$ and paths $\Pi_O$ that satisfy \cref{thm:outer_graph}.

\subsection{Convex Set of Vectors}
\label{subsec:W}

Before specifying the construction of the graph $G_O$, we begin by specifying our construction of a set of vectors $W \subseteq \mathbb{R}^3$ that is crucial to the construction of $G_O$. 
Set $W$ will be parameterized by a positive integer $r$, i.e., $W = W(r)$. 
The vectors in $W$ will satisfy a certain strict convexity property that we will use to ensure the unique shortest paths property of paths $\Pi_O$ in $G_O$.

\begin{definition}[$W(r)$]\label{def:W}
    Given a positive integer $r$, let 
    \[W_1(r) := \{(x,0,x^2) \mid x\in \{r/2,\dots, r\}\} \quad\text{and}\quad W_2(r) := \{(0,y,y^2) \mid y\in \{r/2,\dots, r\}\}.\]
    We define $W(r)$ to be the sumset
    \[W(r) := W_1(r) + W_2(r) = \{(x,y,x^2+y^2)\mid x,y\in \{r/2,\dots, r\}\}. \]
\end{definition}

We now verify that sets of vectors $W_1(r), W_2(r), W(r)$ have the necessary convexity property to ensure that graph $G_O$ has unique shortest paths (Property 2 of \cref{thm:outer_graph}). The convexity property of $W$ stated in \cref{lem:strongly-convex} is roughly similar to the notion of  `strong convexity' in \cite{BodwinH22}, but is in fact stronger. 

\begin{lemma}[Convexity property]\label{lem:strongly-convex}
Let $W_1, W_2, W$ be the sets defined in \cref{def:W} for some positive integer $r$. Let $W'$ be the set
$$
W' = W \cup (-W) \cup (W_1 - W_2) \cup (W_2 - W_1).
$$
Then each vector $\Vec{w} \in W$ is an extreme point of the convex hull $\conv(W')$ of $W'$. 
\end{lemma}
\begin{proof}
Let $W_1 = W_1(r)$, $W_2=W_2(r)$, and $W = W(r)$ for some positive integer $r$. Let $W' = W \cup (-W) \cup (W_1 - W_2) \cup (W_2 - W_1)$. 
Fix a vector $\Vec{w} = (x, y, x^2+y^2) \in W$. Let $\Vec{c}$ be the vector $\Vec{c} = (2x, 2y, -1)$. Then we claim that $\Vec{w}$ is the unique vector in $W'$ such that $\Vec{w} = \max_{\Vec{u} \in W'} \langle \Vec{c}, \Vec{u} \rangle$. Note that this immediately implies that $\Vec{w}$ is an extreme point of $\conv(W')$.  We now verify this claim:
\begin{itemize}
    \item The inner product of $\Vec{c}$ and $\Vec{w}$ is  $\langle \Vec{c}, \Vec{w} \rangle  = x^2 + y^2$.
    \item If $\Vec{u} = (u_1, u_2, u_1^2 + u_2^2) \in W$, then 
    $$
    \langle \Vec{c}, \Vec{u} \rangle = 2xu_1 + 2yu_2 - u_1^2 - u_2^2 = u_1(2x - u_1) + u_2(2y - u_2).
    $$
    It is straightforward to verify that $u_1(2x - u_1) + u_2(2y - u_2)$ is uniquely maximized when $u_1 = x$ and $u_2 = y$ (e.g., using the second partial derivative test). Then $\langle \Vec{c}, \Vec{u} \rangle \leq \langle \Vec{c}, \Vec{w} \rangle$, with equality only if $\Vec{u} = \Vec{w}$.  
    \item  If $\Vec{u} = (-u_1, -u_2, -u_1^2 - u_2^2) \in -W$, then 
    $$
    \langle \Vec{c}, \Vec{u} \rangle = -2xu_1 - 2yu_2 + u_1^2 + u_2^2 = u_1(u_1 - 2x) + u_2(u_2 - 2y) \leq 0 < \langle \Vec{c}, \Vec{w} \rangle,
    $$
    using the fact that $u_1 \leq 2x$ and $u_2 \leq 2y$, since $u_1, x, u_2, y \in [r/2, r]$. 
    \item If $\Vec{u} = \Vec{u}_1 + \Vec{u}_2 \in W_1 - W_2$, where $\Vec{u}_1 \in W_1$ and $\Vec{u}_2 \in -W_2$, then by the previous analyses,
    $$
\langle \Vec{c}, \Vec{u}_2\rangle \leq 0 \text{\qquad and \qquad} \langle \Vec{c}, \Vec{u}_1 \rangle \leq x^2. 
$$
Then $\langle \Vec{c}, \Vec{u}\rangle = \langle \Vec{c}, \Vec{u}_1 \rangle + \langle \Vec{c}, \Vec{u}_2 \rangle \leq x^2 < \langle \Vec{c}, \Vec{w} \rangle$. The case where $\Vec{u} \in W_2 - W_1$ is symmetric. 
\end{itemize}
We have shown that for all $\Vec{u} \in W'$, $\langle \Vec{c}, \Vec{u} \rangle \leq \langle \Vec{c}, \Vec{w} \rangle$, with equality only if $\Vec{u} = \Vec{w}$. 
\end{proof}

\subsection{Construction of $G_O$}
Let $a, r > 0 \in \mathbb{Z}$ be the input parameters for our construction of outer graph $G_O = (V_O, E_O)$. Let $W_1 = W_1(r)$, $W_2 = W_2(r)$, and $W = W(r)$ be the sets of vectors constructed in \cref{def:W} and parameterized by our choice of $r$.

\paragraph{Vertex Set $V_O$.}
\begin{itemize}
    \item Our vertex set $V_O$ will correspond to two copies of integer points arranged in a grid in $\mathbb{R}^3$.
    These two copies will be denoted as $V_O^L$ and $V_O^R$. 
    For a point $p \in \mathbb{R}^3$, we will use $p_L$ to denote the copy of point $p$ in $V_O^L$, and $p_R$ to denote the copy of point $p$ in $V_O^R$. Likewise, for a set of points $P \subseteq \mathbb{R}^3$, we will use $P_L$ to denote the copy of set $P$ in $V_O^L$, and $P_R$ to denote the copy of set $P$ in $V_O^R$.
 Then we define $V_O^L$ and $V_O^R$ as: 
    $$
    V_O^L = ([a] \times [a] \times [ar])_L, \text{ \qquad and \qquad } V_O^R = ([a] \times [a] \times [ar])_R.
    $$
    When denoting a node $v_L$ or $v_R$ in $V_O$, we will drop the subscript and simply denote this node as $v$ when its membership in $V_O^L$ and $V_O^R$ is clear from the context or otherwise irrelevant.

\end{itemize}
\paragraph{Edge Set $E_O$.}
\begin{itemize}
    \item The edges $E_O$ in $G_O$ will pass between $V_O^L$ and $V_O^R$, so that $E_O \subseteq V_O^L \times V_O^R$. 
    \item Just as the nodes in $V_O$ are integer points in $\mathbb{R}^3$, we will identify the edges in $E_O$  with integer vectors in $\mathbb{R}^3$. Specifically, for each  edge $(x_L, y_R)$ in $E_O$,  we identify $x_L \rightarrow y_R$ with the vector $y - x \in \mathbb{R}^3$. Note that $y-x$ corresponds to the orientation $x_L \rightarrow y_R$ of edge $(x_L, y_R)$; we would use vector $x-y \in \mathbb{R}^3$ to denote $y_R\rightarrow x_L$. 
    \item For each node $v_L \in V_O^L$ and each vector $\Vec{w} \in W_1$, if $(v + \vec{w})_R \in V_O^R$, then add edge $(v_L,(v + \vec{w})_R)$ to $E_O$. Likewise, for each node $v_R \in V_O^R$ and each vector $\Vec{w} \in W_2$, if $(v  + \vec{w})_L \in V_O^L$, then add edge $(v_R, (v + \vec{w})_L)$ to $E_O$. 
\end{itemize}

\paragraph{Critical Paths $\Pi_O$.}
\begin{itemize}
    \item Let  $S \subseteq \mathbb{R}^3$ denote the set of points $S = [a] \times [a] \times [r^2/8]$. 
\item 
Let $s$ be a point in $S$. Additionally,  let  $\Vec{w}_1$ be a vector in $W_1$, and let $\vec{w}_2$ be a vector in $W_2$, where $\Vec{w}_1 + \vec{w}_2 \in W$.\footnote{Note that $\Vec{w}_1 + \Vec{w}_2 \in W$ for all $\Vec{w}_1 \in W_1$ and $\Vec{w}_2 \in W_2$, since $W = W_1 + W_2$. However, in \cref{subsec:modifying_W}, we will modify $W$ so that $W$ is a strict subset of $W_1+W_2$, which will make this requirement relevant. \label{footnote:path_def} } If  $s + \vec{w}_1 \not \in S$, then we define a corresponding path in $G_O$ starting from $s_L \in S_L$ as
$$
s_L \rightarrow (s+ \Vec{w}_1)_R \rightarrow (s + \vec{w}_1 + \vec{w}_2)_L \rightarrow   (s + 2\vec{w}_1 + \vec{w}_2)_R \rightarrow   (s + 2\vec{w}_1 + 2\vec{w}_2)_L \rightarrow \dots \rightarrow (s + i \cdot \vec{w}_1 + i \cdot \vec{w}_2)_L,
$$
where $i$ is the largest integer $i$ such that node $(s + i \cdot \vec{w}_1 + i \cdot \vec{w}_2)_L \in V_O^L$. Let $t = (s + i \cdot \vec{w}_1 + i \cdot \vec{w}_2)_L$ be the endpoint of this path, and add this $s \leadsto t$ path to our set of critical paths $\Pi_O$. 
\item Note that every critical path $\pi \in \Pi_O$ constructed this way is uniquely specified by a start node $s_L \in S_L$ and vectors $\Vec{w}_1 \in W_1$ and $\Vec{w}_2 \in W_2$. 

\item For every critical path $\pi \in \Pi_O$ where $|\pi| < \frac{a}{4r}$, remove $\pi$ from $\Pi_O$. 

\end{itemize}
As a final step in our construction of $G_O$, we remove all edges in $G_O$ that do not lie on some critical path $\pi \in \Pi_O$. 

\subsection{Properties of $G_O$}
We will now verify that $G_O$ and $\Pi_O$ satisfy the properties specified in \cref{thm:outer_graph}.

\begin{claim}
\label{claim:outer_graph_unique_path}
Every critical path $\pi \in \Pi_O$ is the unique shortest path between its endpoints in $G_O$. Moreover, $|\pi| \geq \frac{a}{4r}$. 
\end{claim}

\begin{proof}
    Let $\pi \in \Pi_O$ be a critical path with endpoints $s,  t$ that uses vectors $\vec{w}_1\in W_1$ and $\vec{w}_2\in W_2$. By definition we must have $|\pi| \ge \frac{a}{4r}$ since otherwise the path $\pi$ would have been discarded in the construction. It remains to show that $\pi$ is the unique shortest path between its endpoints. Observe that by construction the graph $G_O$ is bipartite and  endpoints $s, t$ both lie in $V_O^L$, so any path from $s$ to $t$ has an even number of edges.

    Suppose for the sake of contradiction that $\pi$ is not a unique shortest path. Then there exists a path $\pi'$, where $\pi' \neq \pi$ and $|\pi'| = |\pi| = 2k$, for some positive integer $k$. Let $\Vec{u}_i$ denote the vector associated with the $i$th edge of $\pi'$, for $i \in [1, 2k]$. Define vector $\Vec{v}_i$ to be $\Vec{v}_i := \Vec{u}_{2i -1} + \Vec{u}_{2i}$ for $i \in [1, k]$. Let $W' = W \cup (-W) \cup (W_1 - W_2) \cup (W_2 - W_1)$. It is straightforward to verify that by  construction, $\Vec{v}_i \in W'$ for $i \in [1, k]$. 

    Since $\pi$ and $\pi'$ are both $s \leadsto t$ paths, it follows that
    $$
    k(\Vec{w}_1 + \Vec{w}_2) = t - s = \sum_{i=1}^k \Vec{v}_i. 
    $$
    Let $\Vec{w} = \Vec{w}_1 + \Vec{w}_2 \in W$. Then this implies that $$\Vec{w} = \frac{1}{k} \sum_{i=1}^k \Vec{v}_i,$$
    where $\Vec{v}_i \in W'$ for all $i \in [1, k]$. Additionally, note that $\Vec{v}_j \neq \Vec{w}$ for some $j \in [1, k]$, since $\pi' \neq \pi$.    However, this implies that $\Vec{w}$ is a convex combination of points in $\Vec{W'} \setminus \{\Vec{w}\}$, contradicting \cref{lem:strongly-convex}.     

\end{proof}

\begin{claim}\label{claim:outer_graph_overlap}
Every pair of distinct critical paths $\pi_1, \pi_2 \in \Pi_O$ can intersect on either a single vertex or a single edge. 
\end{claim}

\begin{proof}
Note that since  critical paths are unique shortest paths between their endpoints by \cref{claim:outer_graph_unique_path}, if two critical paths intersect on two vertices, then they must share the subpath between the two vertices. We show that any pair of critical paths cannot possibly share a subpath of length $2$. 

Let  $\sigma = (x_L, y_R, z_L)$ be a subpath of a critical path $\pi$ of length $2$. We will show that $\pi \in \Pi_O$ is the unique critical path in $\Pi_O$ that contains $\sigma$. The vectors $y - x$ and $z-y$ correspond to two vectors $\vec{w}_1\in W_1$ and $\vec{w}_2\in W_2$. Note that any critical path $\pi'$ in $\Pi_O$ that contains $\sigma$ must take edges corresponding to $\vec{w}_1$ and $\vec{w}_2$ in an alternating fashion. In particular, any path $\pi'$ containing subpath $\sigma$ must have a start node $s_L \in S_L$ where
$$
s_L = x - i \cdot \Vec{w}_1 - i \cdot \Vec{w}_2 \in S_L,
$$
where $i$ is a positive integer. We claim that there is a \textit{unique} positive integer $i$ such that $x - i \cdot \Vec{w}_1 - i \cdot \Vec{w}_2 \in S_L$.  This follows from the fact that the third coordinates of $\Vec{w}_1$ and $\Vec{w}_2$ lie in the range $[r^2/4, r^2]$ by \cref{def:W}, while the third coordinate of any point in $S$ only lies in the range $[r^2/8]$. This implies that there is a unique start node $s_L \in S_L$ shared by all critical paths $\pi'$ in $\Pi_O$ that contain subpath $\sigma$. By construction of $G_O$ and $\Pi_O$, every critical path in $\Pi_O$ is uniquely identified by a start node $s_L \in S_L$ and associated edge vectors $\vec{w}_1\in W_1$ and $\vec{w}_2\in W_2$. Then path $\pi$ is the unique critical path in $\Pi_O$ containing subpath $\sigma$. 
\end{proof}

\begin{claim}\label{claim:outer_graph_edge_use}
Every edge in $G_O$ is used by at most $r/2$ critical paths $\pi \in \Pi_O$.
\end{claim}

\begin{proof}
    This follows immediately from the proof of \cref{claim:outer_graph_overlap}. Given an edge $e$, any critical path $\pi$ that uses the edge $e$ will use the vector $\Vec{e} \in W_1 \cup W_2$ corresponding to $e$. There are $\max(|W_1|, |W_2|) = r/2$ possible choices for the other vector used by $\pi$, which together with the nodes incident to $e$ will uniquely determine a critical path. Thus $e$ can be used by at most $r/2$ critical paths.
\end{proof}

\begin{claim}\label{claim:outer_graph_count}
    The number of nodes, edges, and critical paths in $G_O(a, r)$ is: 
    \begin{align*}
        |V_O| &=  \Theta(a^3 r), \\
        |E_O| &=   \Theta(a^3r^2),\\
        |\Pi_O| &=  \Theta(a^2r^4).
    \end{align*}
\end{claim}
\begin{proof}
We will show $|E_O| = \Theta(a^3 r^2)$ last.
\begin{itemize}
    \item $|V_O|$: Clear from definition.
    
    \item $|\Pi_O|$: Notice that each critical path is uniquely identified by a vertex in $S$ and a pair of vectors $\vec{w}_1\in W_1, \vec{w}_2\in W_2$, where $\Vec{w}_1 + \Vec{w}_2 \in W$, so there are at most $O(|S_L|\cdot |W|) = O(a^2r^4)$ critical paths. On the other hand, although we remove the critical paths $\pi$ of length $|\pi| < a/4r$, we argue that at least a constant fraction of the critical paths remain. To see this, notice that for the vertices in the set $S'=\{1\}\times [a/4]\times [a/4]\times [r^2/8]\subseteq S$ and all possible pairs of vectors $\vec{w}_1\in W_1, \vec{w}_2\in W_2$, the corresponding unique critical path has length at least $a/(4r)$. This proves that $|\Pi_O| = \Theta(a^2r^4)$ as desired.

    \item $|E_O|$: By construction of the edge set, every vertex in $V_O$ has degree at most $|W_1| + |W_2| = O(r)$. Then $|E_O| = O(r \cdot |V_O|) = O(a^3r^2)$. What remains is to show that $|E_O| = \Omega(a^3r^2)$. By \cref{claim:outer_graph_edge_use} and the fact that all critical paths have length at least $a/(4r)$, we have
        \[|E_O|\ge \frac{1}{r}\cdot |\Pi_O|\cdot \frac{a}{4r} = \Omega(a^3 r^2).\]

\end{itemize}

\end{proof}

\begin{proof}[Proof of \cref{thm:outer_graph}]
Note that graph $G_O$ and critical paths $\Pi_O$ satisfy Properties 1, 2, and 3 of \cref{thm:outer_graph} by \cref{claim:outer_graph_unique_path,claim:outer_graph_overlap,claim:outer_graph_count}.  Moreover, Property 4 of \cref{thm:outer_graph} follows immediately from the final step in our construction of $G_O$. This completes the proof of \cref{thm:outer_graph}. 
\end{proof}

\section{Emulator Lower Bound}
\label{sec:emulator_lb}

In this section we will finish our emulator lower bound by constructing the obstacle product  graph $G$ specified in  \cref{thm:emulator}. 

\subsection{Construction of Obstacle Product Graph $G$}

Let $a, r > 0 \in \mathbb{Z}$ be construction parameters to be specified later. Let $G_O = G_O(a, r)$ be an instance of our outer graph with parameters $(a, r)$ from \cref{thm:outer_graph}. We will construct our final graph $G$ by performing the obstacle product. The obstacle product is performed in two steps: the edge subdivision step and the inner graph replacement step. 

\paragraph{Inner Graph $G_I$.}
Our inner graph $G_I(r) = (V_I, E_I)$ will be the biclique $K_{r, r}$. We denote the two sides of the biclique by $L_I = \{x_I^1, \dots, x_I^r\}$ and $R_I = \{y_I^1, \dots, y_I^r\}$.

\paragraph{Edge Subdivision.} We subdivide each edge in $G_O$ into a path of length $\psi = \Theta\left(\frac{a}{r}\right)$. Denote the resulting graph $G_O'$. 
For any edge $e = (u, v) \in E_O$, let $P_e$ denote the resulting $u \leadsto v$ path of length $\psi$. We will refer to the paths  in $G$ replacing edges from $G_O$ as \textit{subdivided paths}. 

\paragraph{Inner Graph Replacement.}
Let $G_I = G_I(r)$ be an instance of our inner graph with input parameter $r$. We perform the following operations on graph $G_O'$.
\begin{itemize}
    \item For each node $v$ in $V(G_O')$ originally in $G_O$, replace $v$ with a copy of $G_I$. We refer to this copy of $G_I$ as $G_I^v$. Likewise, we refer to the partite sets $L_I$ and $R_I$ in $G_I^v$ as $L_I^v$ and $R_I^v$. 
    \item After applying the previous operation, the endpoints of the subdivided paths $P_e$ in $G_O'$ no longer exist in the graph. If $e = (u, v) \in E_O$, then $P_e$ will have endpoints $u$ and $v$. We will replace the endpoints $u$ and $v$ of $P_e$ with nodes in $G_I^u$ and $G_I^v$ respectively. 
    \item In order to precisely define this replacement operation, it will be helpful to  define two injective functions, $\phi_1: W_1 \mapsto L_I \times R_I$ and $\phi_2: W_2 \mapsto L_I \times R_I$. Let $\Vec{w}_1^i$ (respectively, $\Vec{w}_2^i$) denote the $i$th vector in $W_1$ (respectively, $W_2$), for $i \in [1, r/2]$.  Then we define our injective functions to be
    $$
    \phi_1(\Vec{w}_1^i) = (x_I^i, y_I^i) \text{\quad and \quad} \phi_2(\Vec{w}_2^i) = (x_I^i, y_I^i) \text{\quad for $i \in [1, r/2]$.}
    $$
    \item Let $e = (u, v) \in E_O$. If $v-u \in W_1$, then let $\phi_1(v-u) = (x, y) \in L_I \times R_I$. We will replace the endpoints $u$ and $v$ of $P_e$ with nodes $y \in R_I^u$ in $G_I^u$ and $x \in L_I^v$ in $G_I^v$, respectively.  Otherwise, if $v - u \in W_2$, then let $\phi_2(v-u) = (x, y) \in L_I \times R_I$, and replace the endpoints $u$ and $v$ of $P_e$ with nodes $y \in R_I^u$ in $G_I^u$ and $x \in L_I^v$ in $G_I^v$, respectively. We repeat this operation for each $e \in E_O$ to obtain the obstacle product graph $G$.

    \item Note that after performing the previous operation, every subdivided path $P_e$, where $e = (u, v)$, will have a start node in $R_I^u$ and an end node in $L_I^v$. We will use $r_e$ to denote the start node of $P_e$ in $R_I^u$ and $l_e$ to the end node of $P_e$ in $L_I^v$.

\end{itemize}

\paragraph{Critical Paths $\Pi$.}
\begin{itemize}
    \item Fix a critical path $\pi_O \in \Pi_O$ with associated vectors $\Vec{w}_1 \in W_1$ and $\Vec{w}_2 \in W_2$.
 
    \item Let $e_i$ denote the $i$th edge of $\pi_O$ for $i \in [1, k]$.     

    Then we define a corresponding path $\pi$ in $G$:
    $$
    \pi = P_{e_1} \circ (l_{e_1}, r_{e_2}) \circ P_{e_2} \circ \dots \circ P_{e_{k-1}} \circ  (l_{e_{k-1}}, r_{e_{k}}) \circ P_{e_k}
    $$
    Note that if $e_i = (x, y)$ and $e_{i+1} = (y, z)$, then $(l_{e_i}, r_{e_{i+1}})$ corresponds to an edge from $L_I$ to $R_I$ in inner graph copy $G_I^y$. We add path $\pi$ to our set of critical paths $\Pi$.

    \item We repeat this process for all critical paths in $\Pi_O$ to obtain our set of critical paths $\Pi$ in $G$. Each critical path $\pi \in \Pi$ is uniquely constructed from a critical path $\pi_O \in \Pi_O$, so $|\Pi| = |\Pi_O|$. We will use $\phi:\Pi \mapsto \Pi_O$ to denote the bijection between $\Pi$ and $\Pi_O$ implicit in the construction.

\end{itemize}

As the final step in our construction of obstacle product graph $G$, we remove all edges in $G$ that do not lie on some critical path $\pi \in \Pi$. Note that this will only remove edges in $G$ that are inside copies of the inner graph $G_I$.

\subsection{Analysis of $G$.}

\label{subsec:emulator_analysis}

Our stretch analysis follows a similar argument to the emulator stretch analysis in \cite{LuVWX21}. In the analysis, we will crucially use the fact that in our inner graph replacement step, we attach the incoming edges incident to a node $v \in V_O$ to nodes in $L_I^v$, and we attach the outgoing edges incident to $v$ to nodes in $R_I^v$. 
Then the image of any path of length $2$ in $G_O$ will pass through a unique edge in some inner graph copy in $G$.
From now on, we will call the set of edges in the inner graph copies \emph{clique edges} for convenience. Then the key observation is that each critical path in $\Pi$ will pass through a set of unique clique edges. In particular, since a path of length $2$ in $G_O$ uniquely identifies a clique edge, by \cref{claim:outer_graph_overlap} we know that the critical paths in $\Pi$ do not intersect on clique edges. So in the following, we first argue that any spanner on $G$ cannot miss too many clique edges, then we will show that any emulator can effectively be turned into a spanner with the same stretch using some more edges. This will give us a lower bound on the additive stretch of linear-sized emulators.

We begin with a lemma formalizing the above described property of clique edges.

\begin{claim}\label{claim:clique-edge-no-overlap}
    Every pair of distinct critical paths $\pi, \pi' \in \Pi$ do not intersect on clique edges.
\end{claim}

\begin{proof}
    
    Consider the corresponding critical paths $\pi_O, \pi_O'\in \Pi_O$.  If $(x, y, z)$ is a subpath of $\pi_O$, then $\pi$ and $\pi'$ intersect on a clique edge in $G_I^y$ only if $(x, y, z)$ is a subpath of $\pi_O'$. However, by \cref{claim:outer_graph_overlap}, paths $\pi_O, \pi_O' \in \Pi_O$ do not intersect on a path of length $2$, so $(x, y, z) \not \subseteq \pi_O \cap \pi_O'$. 

\end{proof}

\begin{claim}\label{claim:clique-edge-in-crit-path}
    Every critical path $\pi \in \Pi$ contains at least $\frac{a}{4r}-1$ clique edges.
\end{claim}

\begin{proof}
    By construction of $\pi \in \Pi$ from $\pi_O \in \Pi_O$, every node on $\pi_O$ except for the endpoints will correspond to a unique clique edge. Since $|\pi_O|\ge a/4r$, the critical path in $\pi\in \Pi$ constructed from $\pi_O$ will go through at least $\frac{a}{4r} - 1$ clique edges.
\end{proof}

Now we show that any spanner on $G$ must include many clique edges in order to have small additive stretch.

\begin{lemma}\label{lem:emulator-spanner-stretch}
    Any spanner $H\subseteq G$ that contains $< (\frac{a}{8r}-1)\cdot |P|$ clique edges must have additive stretch $\Omega(\min\{\psi,a/r\})$.
\end{lemma}

\begin{proof}
    The proof follows similarly to the proof of \cite[Lemma 6.1]{LuVWX21}. Let $H\subseteq G$ be a spanner containing $< (\frac{a}{8r}-1)\cdot |\Pi|$ clique edges. Then by \cref{claim:clique-edge-no-overlap} and the pigeonhole principle, there exists some $s \leadsto t$ critical path $\pi \in \Pi$ with less than $\frac{a}{8r}-1$ clique edges. By \cref{claim:clique-edge-in-crit-path}, this means that at least $\frac{a}{8r}$ clique edges on $\pi$ are missing in $H$. Let $\pi_H$ denote the shortest path between $s$ and $t$ in $H$ and let $\pi_H^O$ denote the corresponding path in $G_O$ obtained by contracting subdivided paths in $G$ into single edges and contracting inner graph copies $G_I$ in $G$ into single nodes. We compare $\pi_H^O$ against the critical path $\pi_O = \phi^{-1}(\pi) \in \Pi_O$ in $G_O$:

    \begin{itemize}
        \item $\pi_H^O = \pi_O$: This means that $\pi_H$ and $\pi$ traverse the same set of subdivided paths but possibly different clique edges. Furthermore, this implies that  $\pi_H$ and $\pi$ enter and exit inner graph copy $G_I^y$ through nodes $l_{(x, y)}, r_{(y, z)}$, where $(x, y, z) \subseteq \pi_O\cap \pi_H^O$ is a subpath of $\pi_O$ and $\pi_H^O$. For each clique edge $(l_{(x, y)}, r_{(y, z)})$ in $\pi$ missing from $H$, $\pi_H$ must use at least two extra clique edges to go between $l_{(x, y)}$ and $r_{(y, z)}$ in $G_I^y$ because the inner graphs are bipartite. Since there are $\frac{a}{8r}$ clique edges in $\pi$ missing from $H$, we have $|\pi_H|\ge |\pi| + 2\cdot \frac{a}{8r} = |\pi| + \frac{a}{4r}$.

        \item Path $\pi_{H}^O\ne \pi_O$: $\pi_O$ is the unique shortest path between $s$ and $t$ in $G_O$, so $\pi_{H}^O$ must use at least one more edge than $\pi_O$. Since the edges in $G_O$ are subdivided to paths of length $\psi$, we have $|\pi_H|\ge |\pi| + k$.
    \end{itemize}

    This concludes the proof.
\end{proof} 

Then we show that any emulator can be turned into a spanner with the same additive stretch using some more edges.

\begin{lemma}\label{lem:emulator-stretch}
    Any emulator $H$ of $G$ containing $< |\Pi|/32$ edges must have additive stretch $\Omega(\min\{\psi,a/r\})$.
\end{lemma}

The proof follows almost identically from \cite[Lemma 6.2]{LuVWX21}, but we repeat the proof for completeness.

\begin{proof}
    Let $H$ be such an emulator for $G$. We will construct a spanner $H'\subseteq G$ with similar additive stretch. Note that we may assume without loss of generality that every edge $(u,v)\in E(H)$ has weight $\wt(u,v) = \dist_{G}(u,v)$. 
   For each edge $(u, v) \in E(H)$ such that $u, v \in \pi$ for some critical path $\pi$ in $\Pi$, we add a $u \leadsto v$ path $\pi'$ to $H'$, where $|\pi'| = \dist_G(u, v)$ and $\pi' \subseteq G$. Then by construction, for all $(u, v) \in E(H)$ such that $u, v \in \pi$ for some $\pi \in \Pi$, $\dist_H(u, v) = \dist_{H'}(u, v)$. 

    Now we compute the number of clique edges in $H'$. We claim that each edge $(u, v)$ in $E(H)$ such that $u, v \in \pi$ for some $\pi \in \Pi$ can contribute at most $2a/r$ clique edges to $H'$. Indeed, the number of clique edges contributed by an edge $(u, v) \in E(H)$ to $H'$ is at most the number of clique edges in $\pi \in \Pi$. This is at most the length of the path $\pi_O = \phi^{-1}(\pi)$,  the path in $\Pi_O$ associated with $\pi$. Since $|\pi_O| \leq 2a/r$, the claim follows.

    Thus the number of clique edges in $H'$ is $< \frac{2a}{r}\cdot \frac{|\Pi|}{32} < (\frac{a}{8r}-1)\cdot |\Pi|$. Then by \cref{lem:emulator-spanner-stretch}, $H'$ must have stretch $\Omega(\min\{\psi,a/r\})$, so $H$ must have stretch $\Omega(\min\{\psi,a/r\})$ as well.
\end{proof}

We are now finally ready to proof \cref{thm:emulator}.

\begin{proof}[Proof of \cref{thm:emulator}]
    We set $\psi = \Theta(a/r)$. Then the graph $G$ satisfies:
    \begin{align*}
        |V| &= \Theta(\psi \cdot |E_O| + r\cdot |V_O|) = \Theta(a^4r + a^3r^2)\\
        |\Pi| &= |\Pi_O| = \Theta(a^2r^4).
    \end{align*}
    We set $|\Pi| = \Theta(|V|)$ so we get that $a = \Theta(r^{3/2})$. Then we have $|V| = \Theta(r^7)$. By \cref{lem:emulator-stretch}, any emulator on $\Theta(|V|)$ edges must have stretch $\Omega(\min\{\psi, a/r\}) = \Omega(a/r) = \Omega(r^{1/2})$. Since $n = |V| = \Theta(r^7)$, we can conclude that any linear-sized emulator on $G$ with $n$ vertices must have additive stretch $\Omega(n^{1/14})$.
\end{proof}
\section{Spanner Lower Bound Construction}

% \ghnoteinline{Section 6 is stable. Please read through it and add comments or changes at your leisure.}

In this section we present our lower bound construction for additive spanners. This construction will have a similar structure to the obstacle product graph $G$ constructed for our emulator lower bound in \cref{sec:emulator_lb}, but with several complications. We now describe these  modifications to the obstacle product argument:
\begin{itemize}
    \item \textbf{Convex Sets $W_1, W_2,$ and $W$.} In \cref{subsec:modifying_W}, we modify the convex sets of vectors $W_1, W_2,$ and $W_3$ that we defined in \cref{sec:outer_graph}. 
    The purpose of this modification is technical, but it has to do with the projection argument we employ in our analysis in \cref{sec:spanner_analysis}. Our new convex sets of vectors $W_1(r, c), W_2(r, c), W(r, c)$ will now be parameterized by an additional integer $c > 0$. These new sets of vectors  will roughly resemble the outer graph vectors in Lemma 8 of \cite{BodwinH22} and  will play a similar role in our analysis of $G$. 
    \item \textbf{Inner Graph $G_I$.} Instead of choosing our inner graphs to be bicliques as in \cref{sec:emulator_lb}, we will choose our inner graphs to be the sourcewise distance preserver lower bound graphs constructed in \cite{CoppersmithE06}. \cref{lem:sourcewise_dp} specifies the exact properties of these new inner graphs $G_I$ that we require in our analysis. See Subsections \ref{subsec:program} and \ref{subsec:overview-outer} for an overview of why we make this design choice. 
\end{itemize}

\subsection{Modifying Convex Sets $W_1, W_2, \text{ and } W$ in $G_O$}
\label{subsec:modifying_W}

In this section, we modify the definitions of the convex sets of vectors $W_1, W_2,$ and $W$ used to construct outer graph $G_O$. Let $r, c > 0 \in \mathbb{Z}$ be the input parameters to our convex sets of vectors $W_1, W_2$, and $W$.

 We define $I_i$ be the interval $$I_i := \left[\frac{r}{2} + \frac{(2i-2)\cdot r}{4c}, \quad \frac{r}{2} + \frac{(2i-2) \cdot r }{4c} + \frac{r}{16c^3}  \right],$$ for $i \in [1, c]$. Our specific choice of intervals $I_i$ will be come relevant in the proofs of Claims \ref{claim:backward} and \ref{claim:sink}.
The following claim is immediate from the definition of intervals $I_i$.

\begin{claim}
\label{claim:intervals}
    Intervals $\{I_i\}_{i \in [1, c]}$ satisfy the following properties:
    \begin{itemize}
        \item $I_i \subseteq [r/2, r]$,
        \item $|I_i| = \frac{r}{16c^3}$
        \item if $x, y \in I_i$, then $|x - y| \leq \frac{r}{16c^3}$, and
        \item if $x \in I_i$ and $y \in I_j$, where $i \neq j$, then $|x-y| \geq r/(2c)$. 
    \end{itemize}
\end{claim}

We will use intervals $\{I_i\}_{i \in [1, c]}$ to construct our sets of vectors $W_1(r, c)$ and $W_2(r, c)$. 

\begin{definition}[$W_1(r, c)$ and $W_2(r, c)$]
    Let $r, c$ be  positive integers. We define $W_1(r, c)$ and $W_2(r, c)$ as
    \[W_1(r, c) := \left\{(x, 0, x^2) \mid x\in I_i, i \in [1, c]\right\} \text{\quad and \quad} W_2(r, c) := \left\{(0, y, y^2) \mid y\in I_i, i \in [1, c]\right\}. \]
\end{definition}

Now we partition the vectors in $W_1(r, c)$ into $c$ sets $\mathcal{S}^1_1, \dots, \mathcal{S}^1_c$ we call \textit{stripes}. We define the $i$th stripe $\mathcal{S}^1_i$ of $W_1(r, c)$ as $\{(x, 0, x^2) \mid x \in I_i\}$. Likewise, we define the $i$th stripe $\mathcal{S}^2_i$ of $W_2(r, c)$ as 
$\{(0, y, y^2) \mid y \in I_i\}$. 
The key properties of our stripes are summarized in \cref{claim:stripes}, which follows immediately from \cref{claim:intervals}.

\begin{claim}
    \label{claim:stripes}
    Stripes $\{\mathcal{S}_i^1\}_{i \in [1,c]}$ satisfy the following properties:
    \begin{itemize}
        \item $\mathcal{S}_i^1 \subseteq [r/2, r]$, 
        \item $|\mathcal{S}_i^1| = \frac{r}{16c^3}$, 
        \item if $(x, 0, x^2), (y, 0, y^2) \in \mathcal{S}_i^1$, then $|x-y| \leq \frac{r}{16c^3}$, and
        \item if $(x, 0, x^2) \in \mathcal{S}_i^1$ and $(y, 0, y^2) \in \mathcal{S}_j^1$, where $i \neq j$,  then $|x - y| \geq \frac{r}{2c}$.
    \end{itemize}
    Moreover, stripes $\{\mathcal{S}_i^2\}_{i \in [1,c]}$ satisfy analogous properties.
\end{claim}

Roughly, \cref{claim:stripes} states that vectors in the same stripe in $W_1(r, c)$ or $W_2(r, c)$ are ``close'' to each other in some sense, and vectors in different stripes in $W_1(r, c)$ and $W_2(r, c)$ are ``far'' from each other in some sense. This notion of partitioning a set of vectors into stripes satisfying these properties was introduced in the spanner lower bound construction of \cite{BodwinH22}. 
We are now ready to define our set of vectors $W(r, c)$. 

\begin{definition}[$W(r, c)$] Let $r, c$ be positive integers. Unlike in Definition \ref{def:W}, we will define $W(r, c)$ to be a \textit{subset} of the sumset $W_1(r, c) + W_2(r, c)$.
In particular, for $\vec{w}_1 \in W_1(r, c)$ and $\vec{w}_2 \in W_2(r, c)$, we only add $\vec{w}_1 + \vec{w}_2$ to $W(r, c)$ if $\vec{w}_1$ and $\vec{w}_2$ share the same stripe index $i \in [1, c]$. 
Formally,
    \[
    W(r, c) := \left\{  (x, y, x^2 + y^2) \mid   x, y \in I_i, i \in [1, c]\right\} \subset W_1(r, c) + W_2(r, c).
    \]
\end{definition}
The following claim is immediate from the definitions of $W_1(r, c)$, $W_2(r, c)$ and $W(r, c)$. 

\begin{claim}
    Sets $W_1(r, c)$, $W_2(r, c)$ and $W(r, c)$ satisfy the following properties:
    \begin{itemize}
        \item $|W_1(r, c)| = |W_2(r, c)| = \Theta\left( \frac{r}{c^2} \right)$,
        \item $|W(r, c)| = \Theta\left( \frac{r^2}{c^5}\right)$, 
        \item $W_1(r, c) \subset W_1(r)$, $W_2(r, c) \subset W_2(r)$, and $W(r, c) \subset W(r)$, so $W(r, c)$ satisfies the convexity property stated in Lemma \ref{lem:strongly-convex}. 
    \end{itemize}
\end{claim}

We modify the construction of outer graph $G_O$ in \cref{sec:outer_graph} by replacing sets $W_1(r), W_2(r)$, and $W(r)$ defined in \cref{subsec:W} with the new sets $W_1(r, c), W_2(r, c)$, and $W(r, c)$. Note that our new choice of sets $W_1(r, c)$ and $W_2(r, c)$ changes the the set of vectors $E_O$ in $G_O$, while our new choice of set $W(r, c)$ changes the set of critical paths $\Pi_O$ (see \cref{footnote:path_def}).

By inserting convex sets $W_1(r, c), W_2(r, c)$, and $W(r, c)$ into $G_O$ in place of the sets $W_1(r), W_2(r)$, and $W(r)$, we obtain the following theorem about our modified outer graph $G_O = G_O(a, r, c)$.

\begin{theorem}[Properties of Modified Outer Graph]
\label{thm:modified_outer_graph}
For any $a, r, c > 0 \in \mathbb{Z}$, there exists a graph $G_O(a, r, c) = (V_O, E_O)$ with a set of critical paths $\Pi_O$ that has the following properties:
\begin{enumerate}
    \item The number of nodes, edges, and critical paths in $G_O$ is:
    \begin{align*}
        |V_O| &=  \Theta(a^3 r), \\
        |E_O| &=   \Theta\left(\frac{a^3r^2}{c^2}\right),\\
        |\Pi_O| &=  \Theta\left(\frac{a^2r^4}{c^5}\right).
    \end{align*}
    \item Every critical path $\pi \in \Pi_O$ is a unique shortest path  in $G_O$ of length at least $|\pi| \geq \frac{a}{4r}$.
    \item Every pair of distinct critical paths $\pi_1$ and $\pi_2$ intersect on at most two  nodes. 
    \item Every edge $e \in E_O$ lies on some critical path in $\Pi_O$. 
\end{enumerate}
\end{theorem}

Just like in the original construction of $G_O$ in \cref{sec:outer_graph}, every critical path $\pi \in \Pi_O$ corresponds to a unique vector $\Vec{w} \in W(r, c)$. Specifically, by the definition of $W_1(r, c), W_2(r, c),$ and $W(r, c)$,  path $\pi$ is constructed using vectors $\Vec{w}_1 \in W_1(r, c)$ and $\vec{w}_2 \in W_2(r, c)$, where
\begin{itemize}
    \item $\vec{w} = \vec{w}_1 + \Vec{w}_2$, and
    \item $\Vec{w}_1 \in \mathcal{S}_i^1$ and $\Vec{w}_2 \in \mathcal{S}_i^2$, for some $i \in [1, c]$.
\end{itemize}
Critically, $\Vec{w}_1$ and $\Vec{w}_2$ both lie in the $i$th stripe $\mathcal{S}_i^1$ and $\mathcal{S}_i^2$, respectively. 

\subsection{Inner Graph $G_I$}
\label{subsec:inner_graph}

In this subsection, we formally state the properties of the family of graphs we  choose for our inner graphs $G_I$ when constructing the obstacle product graph $G$. We will choose our inner graphs to be the sourcewise distance preserver lower bound graphs constructed in \cite{CoppersmithE06}. \cref{lem:sourcewise_dp} formally captures the exact properties of this family of graphs that we need for spanner lower bound argument. We will defer our proof of \cref{lem:sourcewise_dp} to \cref{app:sslb}, as it largely follows from the proof of Theorem 5.10 in \cite{CoppersmithE06}.

\begin{lemma}[cf. Theorem 5.10 of \cite{CoppersmithE06}] 
\label{lem:sourcewise_dp}
For any $a, c > 0 \in \mathbb{Z}$, there exists a graph $G_I(a, c) = (V_I, E_I)$ with a set $S_I \subseteq V_I$ of sources, a set $T_I \subseteq V_I$ of sinks, and a set $P_I \subseteq S_I \times T_I$ of critical pairs that has the following properties:
\begin{enumerate}
\item The number of nodes, edges,  sources, sinks, and critical pairs in $G_I$ is:
\begin{align*}
    |V_I| &= \Theta(a^2), \\
    |E_I| &= \Theta(a^2c), \\
    |S_I| &= \Theta(a^{1/2}c^{11/4}), \\
    |T_I| &= \Theta(a^{1/2}c^{11/4}), \\
        |P_I| &= \Theta(ac^{5/2}).
\end{align*}
\item Every path $\pi_{s, t}$, where $(s, t) \in P_I$, contains $\Theta(a/c^{3/2})$ edges that do not lie on any other path $\pi_{s', t'}$, where $(s', t') \in P_I$. 
\item For every source $s \in S_I$ and sink $t \in T_I$, the distance between $s$ and $t$ in $G_I$ satisfies the following:
$$
 \dist_{G_I}(s, t) = \Theta(ac^{1/2}).
$$

\item The set of sources $S_I$ can be partitioned into $b = \Theta(c^3)$ sets $S_I^1, \dots, S_I^b$, where $|S_I^i| = \Theta(a^{1/2}c^{-1/4})$ for all $i \in [b]$. Let $T_I^i = \{t \in T_I \mid (S_I^i \times \{t\}) \cap P_I \neq \emptyset \}$ be the set of all sinks that belong to a critical pair with a source in $S_I^i$. Then for all $i \in [b]$  the following properties hold:
\begin{itemize}
    \item $|T_I^i| = \Theta(a^{1/2}c^{-1/4})$ for all $i \in [b]$,
    \item $S_I^i \times T_I^i \subseteq P_I$, and
    \item for all  $(s, t) \in P_I$ such that $s \in S_I^i$ and $t \in T_I^i$, 
$$
\dist_{G_I}(s, t) \leq \dist_{G_I}(S_I^i, T_I^i),$$
where $\dist_{G_I}(S_I^i, T_I^i)$ denotes the minimum distance between $S_I^i$ and $T_I^i$ in $G_I$. 
\end{itemize}
\end{enumerate}
\end{lemma}
\begin{proof}
We defer the proof of this statement to \cref{app:sslb}. 
\end{proof}

\subsection{Construction of Obstacle Product Graph $G$}
Let $a, r, c > 0 \in \mathbb{Z}$ be the input parameters of an instance of outer graph $G_O = G_O(a, r, c)$. Let $W_1 = W_1(r, c)$, $W_2 = W_2(r, c)$, and $W = W(r, c)$ be the sets of vectors constructed in \cref{subsec:modifying_W}. 
Additionally, let $a', c' > 0 \in \mathbb{Z}$ be the input parameters of an instance of inner graph $G_I = G_I(a', c')$. We will specify the precise values of $a$, $r$, $c$, $a'$, and $c'$ later, as needed. Roughly, our choices of parameters $a$, $r$, and $a'$ will grow with the size of our final graph $G$, while parameters $c$ and $c'$ will be (sufficiently large) integer constants.

We will construct our final graph $G$ by performing the obstacle product. The obstacle product is performed in two steps: the edge subdivision step and the inner graph replacement. In the inner graph replacement step, we will need to carefully define 
two   functions, $\phi_1: W_1 \mapsto S_I \times T_I$ and $\phi_2: W_2 \mapsto S_I \times T_I$ between   vectors in $W_1$ and $W_2$ and pairs of nodes in $S_I \times T_I$ in inner graph $G_I$.

\paragraph{Edge Subdivision.}
We subdivide each edge in $G_O$ into a path of length $\psi$. Denote the resulting graph as $G_O'$. For any edge $e = (u, v) \in E_O$, let $P_e$ denote the resulting $u \leadsto v$ path of length $\psi$. We will take $\psi = \Theta\left( \frac{r^3}{c^{29/3}} \right)$. 

\paragraph{Inner Graph Replacement.}
We now perform the inner graph replacement step of the obstacle product. 
% For ease of understanding, we will first assume the existence of a bijection $\phi: W(r, c) \mapsto P_I$ between vectors in $W(r, c)$ and critical pairs in $P_I$ with a certain useful property.
\begin{itemize}
    \item For each node $v$ in $V(G_O')$ originally in $G_O$, replace $v$ with a copy of $G_I$. We refer to this copy of $G_I$ as $G_I^v$. Likewise, refer to the sources and sinks $S_I$ and $T_I$ in $G_I^v$ as $S_I^v$ and $T_I^v$. 
    \item After applying the previous operation, the endpoints of the subdivided paths $P_e$ in $G_O'$ no longer exist in the graph. If $e = (u, v) \in E_O$, then $P_e$ will have endpoints $u$ and $v$. In the next step, we will replace the endpoints $u$ and $v$ of $P_e$ with nodes in $G_I^u$ and $G_I^v$, respectively.
    % \ghnoteinline{Gary edit barrier}
    \item In order to precisely define this replacement operation, it will be helpful to  define two   functions, $\phi_1: W_1 \mapsto S_I \times T_I$ and $\phi_2: W_2 \mapsto S_I \times T_I$. For ease of understanding, we will first assume the existence   functions $\phi_1$ and $\phi_2$. We will specify our choices of $\phi_1$ and $\phi_2$ later.

    \item Let $e = (u, v) \in E_O$. If $v-u \in W_1$, then let $\phi_1(v-u) = (x, y) \in S_I \times T_I$. We will replace the endpoints $u$ and $v$ of $P_e$ with nodes $y \in T_I^u$ in $G_I^u$ and $x \in S_I^v$ in $G_I^v$, respectively.  Otherwise, if $v - u \in W_2$, then let $\phi_2(v-u) = (x, y) \in S_I \times T_I$, and replace the endpoints $u$ and $v$ of $P_e$ with nodes $y \in T_I^u$ in $G_I^u$ and $x \in S_I^v$ in $G_I^v$, respectively. We repeat this operation for each $e \in E_O$ to obtain the obstacle product graph $G$. 
    % \ghnoteinline{a diagram of how edges are attached may be helpful here.}

    \item Note that after performing the previous operation, every subdivided path $P_e$, where $e = (u, v)$, will have a start node in $T_I^u$ and an end node in $S_I^v$. We will use $t_e$ to denote the start node of $P_e$ in $T_I^u$ and $s_e$ to the end node of $P_e$ in $S_I^v$.
    % \ghnoteinline{Gary edit barrier}
    % \item Now for each node $v$ in $V_O$, we need to reattach the edges originally incident to $v$ in $G_O$ to nodes in $G_I^v$. Let $E_{in}$ denote the incoming edges incident to $v$ in $G_O$, and let $E_{out}$ denote the outgoing edges. 
    % \item For ease of understanding, we will first assume the existence of a bijection $\phi_{in}:E_{in} \mapsto S_I$ between incoming edges in $E_{in}$ and sources in $S_I$. Likewise, we will assume the existence of a bijection $\phi_{out}:E_{out} \mapsto T_I$ between outgoing edges in $E_{out}$ and sinks in $T_I$. Bijections $\phi_{in}$ and $\phi_{out}$ will be specified later.
    % \item For each edge $e \in E_{in}$, reattach path $P_e$ to source $\phi_{in}(e) \in S_I$ in $G_I^v$. Likewise, for each edge $e \in E_{out}$, reattach path $P_e$ to sink $\phi_{out}(e) \in T_I$ in $G_I^v$. 
\end{itemize} 
This completes the construction of the obstacle product graph $G$ (up to defining   functions $\phi_1$ and $\phi_2$).

% In order to ensure that our desired    functions $\phi_1$ and $\phi_2$ exist, we require that $|S_I| = \Theta(r)$ and $|T_I| = \Theta(r)$, so $a' = \Theta\left(  \frac{r^2}{c'^3}  \right)$. 

% \ghnoteinline{Need to explain how to handle paths that do not pass through this inner graph.}

\paragraph{Defining   functions $\phi_1$ and $\phi_2$.} 
Let $\mathcal{S}_1^1, \dots, \mathcal{S}_c^1$ be the stripes of $W_1$, and let  $\mathcal{S}_1^2, \dots, \mathcal{S}_c^2$ be the stripes of $W_2$.  Let $S_I^1, \dots, S_I^b$ and $T_I^1, \dots, T_I^b$ be the partition of sources $S_I$ and sinks $T_I$ as described in Lemma \ref{lem:sourcewise_dp}, where $b = \Theta(c'^3)$. In order to construct our desired   functions, we will require the following relations to hold:
\begin{itemize}
    \item $b \geq c$,
    \item $|S_I^i| \geq |\mathcal{S}_i^j|$, for all $i \in [1, c]$ and $j \in \{1, 2\}$, and 
    \item $|T_I^i| \geq |\mathcal{S}_i^j|$, for all $i \in [1, c]$ and $j \in \{1, 2\}$.
\end{itemize} 
This can be achieved by setting
$$c' = \Theta(c^{1/3}) \text{\quad and \quad} a' = \Theta\left( \frac{r^2}{c^{35/6}} \right),$$
using the fact that $b = \Theta(c'^3)$,  $|\mathcal{S}_i^j| = \Theta(r/c^3)$, $|S_I^i| = \Theta(a'^{1/2}c'^{-1/4})$, and $|T_I^i| = \Theta(a'^{1/2}c'^{-1/4})$ by \cref{claim:stripes} and \cref{lem:sourcewise_dp}. 

We are now ready to define our   functions $\phi_1$ and $\phi_2$. Let $\Vec{w}^k_{i, j}$ denote the $j$th vector of $\mathcal{S}^k_i$, where $i \in [1, c]$,   $j \in [1, |\mathcal{S}^k_i|]$,  and $k \in \{1, 2\}$. Let $s^i_j$ denote the $j$th node of $S_I^i$, where $i \in [1, c]$ and $j \in [1, |S_I^i|]$. Likewise, let $t^i_j$ denote the $j$th node of $T_I^i$, where $i \in [1, c]$ and $j \in [1, |T_I^i|]$. 
We define $\phi_1$ and $\phi_2$ as follows:
$$
\phi_1(\Vec{w}^1_{i, j}) = (s_j^i, t_j^i) \text{\qquad and \qquad} \phi_2(\Vec{w}^2_{i, j}) = (s_j^i, t_j^i) \text{\qquad for $i \in [1, c]$, $j \in [1, |\mathcal{S}^1_j|]$}.
$$
The key properties of functions $\phi_1$ and $\phi_2$ are summarized in \cref{claim:functions}.

\begin{claim}
    \label{claim:functions}
    Functions $\phi_1$ and $\phi_2$ satisfy the following properties:
    \begin{enumerate}
        % \item Let $\Vec{w}, \Vec{w}' \in W_{k}$, $k \in \{1, 2\}$, and let $\phi_k(\Vec{w}) = (s, t)$ and $\phi_k(\Vec{w}') = (s', t')$. Then $s = s'$ or $t = t'$ only if $\Vec{w} = \Vec{w}'$. 
        \item Our choice of $\phi_1$ and $\phi_2$ imply that for each node $u \in S_I \cup T_I$ in an inner graph copy $G_I^v$, there is at most one subdivided path $P_e$ incident to $u$ in $G$. 
        \item $\phi_k(\mathcal{S}_i^1) \subseteq S_I^i \times T_I^i$ for $k \in \{1, 2\}$ and $i \in [1, c]$, where $\phi_k(\mathcal{S}_i^1)$ denotes the image of $\mathcal{S}_i^1$ under $\phi_k$. 
    \end{enumerate}
\end{claim}
\begin{proof}
Note that Property 2 is immediate from the definition of $\phi_1$ and $\phi_2$. What remains is to prove Property 1. Note that $\phi_1$ and $\phi_2$ are injective functions. Each pair $(s_j^i, t_j^i) \in S_I \times T_I$ is uniquely determined by vector $\Vec{w}^i_j$. 
 Let $\Vec{w}, \Vec{w}' \in W_{k}$, $k \in \{1, 2\}$. Then $\phi_k(\Vec{w}) = (s^i_j, t^i_j)$ and $\phi_k(\Vec{w}') = (s^{i'}_{j'}, t^{i'}_{j'})$, for $i, i' \in [1, c]$ and  $j, j' \in [1, |\mathcal{S}_j^k|]$. If $s_j^i = s_{j'}^{i'}$ or $t_j^i = t_{j'}^{i'}$, then $i = i'$ and $j = j'$. This implies that $\Vec{w}  = \Vec{w}'$ since $\phi_k$ is injective. By the construction of graph $G$, we conclude that nodes $s_j^i, t_j^i \in S_I \cup T_I$ are each incident to at most one subdivided path $P_e$.  
\end{proof}

\paragraph{Critical Paths $\Pi$.} 

\begin{itemize}
    \item Fix a critical path $\pi_O \in \Pi_O$ with associated vectors $\Vec{w}_1 \in W_1$ and $\Vec{w}_2 \in W_2$.
    % \item Let $\phi_1(\Vec{w}_1) = (x_1, y_1)$, and let $\phi_2(\Vec{w}_2) = (x_2, y_2)$. 
    \item By our construction of $G_O$, there exists $\Vec{w} \in W$ such that $ \Vec{w} = \Vec{w}_1 + \Vec{w}_2$ (see \cref{footnote:path_def}). Then by the construction of $W$ there exists some index $i \in [1, c]$ such that every edge $(u, v) \in \pi_O$ satisfies $v - u \in \{\Vec{w}_1, \Vec{w}_2\} \subseteq  \mathcal{S}_i^1 \cup \mathcal{S}_i^2$. Let $\chi \in [1, c]$ denote this index. 
    \item Let $e_i$ denote the $i$th edge of $\pi_O$ for $i \in [1, k]$.   
    Note that by Property 2 of \cref{claim:functions}, it follows that $s_{e_i} \in S_I^{\chi}$ and $t_{e_i} \in T_I^{\chi}$ for $i \in [1, k]$. Then by Property 5 of \cref{lem:sourcewise_dp},  $(s_{e_i}, t_{e_{i+1}}) \in P_I$. By Property 2 of \cref{lem:sourcewise_dp}, path $\pi_{s_{e_i}, t_{e_{i+1}}}$ is a unique shortest $s_{e_i} \leadsto t_{e_{i+1}}$ path in $G_I$. 
    % $\pi_O = (v_1, \dots, v_{2k})$. By construction, $v_{2i-1} \in V_O^L$ and $v_{2i} \in V_O^R$, for $i \in [1, k]$.
    \item We now define a  corresponding path $\pi$ in $G$:
    $$
    \pi = P_{e_1} \circ \pi_{s_{e_1}, t_{e_2}} \circ P_{e_2} \circ \dots \circ P_{e_{k-1}} \circ  \pi_{s_{e_{k-1}}, t_{e_{k}}} \circ P_{e_k},
    $$
    % where $\pi_{s_{e_{i-1}}, t_{e_{i}}}$ denotes the unique shortest path 
    Note that if $e_i = (x, y)$ and $e_{i+1} = (y, z)$, then $\pi_{s_{e_i}, t_{e_{i+1}}}$ corresponds to the unique shortest path between $s_{e_i} \in S_I^{\chi}$ and $t_{e_i} \in T_I^{\chi}$ in inner graph copy $G_I^y$. We add path $\pi$ to our set of critical paths $\Pi$. 
    
    % \item Let $s'$ be the vertex in inner graph copy $G_I^{s}$ incident to the outgoing edge corresponding to vector $\vec{w}_1$. Likewise, let $t'$ be the vertex in inner graph copy $G_I^{t}$ incident to the incoming edge corresponding to vector $\vec{w}_2$. We add  $(s', t')$ to our set of critical pairs $P_O$. 
    % \item The  critical path $\pi_{s', t'}$ associated with $(s', t')$ is constructed as follows. Let $\pi_{s, t} \in \Pi_O$ be the critical path associated with $(s, t) \in P_O$. Let $e_1, \dots, e_k$ be the sequence of edges occurring in $\pi_{s, t}$. Now replace each edge $e_i \in E_O$ with the associated subdivided path $P(e_i) \subseteq G$. Then the resulting path, $P(e_1) \circ \dots \circ P(e_k)$, is indeed a path in $G$ because the ending node of $P(e_i)$ will be incident to the starting node of $P(e_{i+1})$ via a biclique edge in some inner graph copy. This path will be the critical path, $\pi_{s', t'}$, associated with critical pair $(s', t') \in P$. 

    \item We repeat this process for all critical paths in $\Pi_O$ to obtain our set of critical paths $\Pi$ in $G$. Each critical path $\pi \in \Pi$ is uniquely constructed from a critical path $\pi_O \in \Pi_O$, so $|\Pi| = |\Pi_O|$. We will use $\phi:\Pi \mapsto \Pi_O$ to denote the bijection between $\Pi$ and $\Pi_O$ implicit in the construction.

    % Each critical path $\pi_O \in \Pi_O$ corresponds to a distinct critical path $\pi \in \Pi$.     
    % We define $\phi:\Pi_O \mapsto \Pi$ to be the bijection between $\Pi_O$ and $\Pi$ implicit in our construction of $\Pi$. 
    
    % \item We repeat this process for all critical pairs $(s, t) \in P_O$, so $|P| = |P_O|$. 

    % \item Since each critical path in $\Pi$ is uniquely constructed from a critical path in $\Pi_O$, given $\pi_{s,t}\in \Pi$, we use $\pi_{s,t}^O$ to denote the corresponding critical path in $\Pi_O$.
   
\end{itemize}

As the final step in our construction of obstacle product graph $G$, we remove all edges in $G$ that do not lie on some critical path $\pi \in \Pi$. Note that this will only remove edges in $G$ that are inside copies of the inner graph $G_I$. \cref{thm:spanner_obstacle_graph} summarizes some of the key properties of obstacle product graph $G$. The proof of \cref{thm:spanner_obstacle_graph} follows from straightforward calculations and arguments similar to those in \cref{subsec:emulator_analysis}.

% \ghnoteinline{edit barrier}

% \begin{itemize}
%     \item Fix a critical path $\pi_O \in \Pi_O$ with associated vectors $\Vec{w}$
% \end{itemize}

% Each path in $\Pi$ will be an image of a path in $\Pi_O$. Fix a path $\pi \in \Pi_O$ with nodes $v_1, \dots, v_k$ and edges $e_1, \dots, e_{k-1}$.
% Note that since $\pi$ is a path in $P_O$, it follows that there exists a $j \in [1, c]$ such that all edges in $\pi$ belong to the $j$th stripe in $W_1(r, c)$ or $W_2(r, c)$, i.e., $e_i \in \mathcal{S}_j^1 \cup \mathcal{S}_j^2$ for all $i \in [1, k-1]$. 
% \begin{itemize}
%     \item We will replace each edge $e_i \in \pi$ with the subdivided path $P(e_i)$.
%     \item  For each node $v_i \in \pi$, $1 < i < k$, let $s_i = \phi_{in}(e_i) \in S_I^j$ in $G_I^{v_i}$, and let $t_i = \phi_{out}(e_{i+1}) \in T_I^j$ in $G_I^{v_i}$. Note that $(s_i, t_i) \in P_I$ by \cref{?}. Then there exists a unique shortest $s_i \leadsto t_i$ path $\pi_{s_i, t_i}$ in $G_I^{v_i}$, again by \cref{?}. We will replace node $v_i$ with the path $\pi_{s_i, t_i}$ in $G_I^{v_i}$. 
%     \item We will ignore nodes $v_1$ and $v_k$ in $\pi$. The resulting path $\pi'$ in the obstacle product graph $G$ will be
% $$
% \pi' =  \pi_{s_1, t_1} \circ P(e_2) \circ \dots \circ P(e_{k-2}) \circ \pi_{s_{k-1}, t_{k-1}} \circ P(e_{k-1}).
% $$
% \item We add  path $\pi'$ to $\Pi$, and repeat this process for all $\pi \in \Pi_O$. 
% \end{itemize}

\begin{theorem}[Properties of Obstacle Product Graph]
\label{thm:spanner_obstacle_graph}
For any $a, r, c > 0 \in \mathbb{Z}$, there exists a graph $G(a, r, c) = (V, E)$ with a set of critical paths $\Pi$ that has the following properties:
\begin{enumerate}
    \item The number of nodes, edges, and critical paths in $G$ is:
    \begin{align*}
        |V| &  = \Theta(a^3r^5c^{-23/2}), \\
        |E| & =   \Theta\left(c^{1/4} \cdot |V|\right),\\
        |\Pi| &=  \Theta\left(\frac{a^2r^4}{c^5}\right).
    \end{align*}
    % \item Every critical path in $\Pi$ is the image of a critical path in $\Pi_O$ in graph $G_O$. 
    % \item Every edge $e \in E$ lies on some critical path in $\Pi$. 
    \item Every path $\pi \in \Pi$ that passes through inner graph copy $G_I^v$ contains $\Theta(a'/c'^{3/2})$ edges that do not lie on any other path $\pi' \in \Pi$, for all $v \in V_O$. 
    
    % Every pair of distinct critical paths $\pi, \pi' \in \Pi$ do not intersect on any inner graph edges in $E(G_I^v)$, for any $v \in V_O$. 
\end{enumerate}
\end{theorem}
\begin{proof} We will prove the size bounds of $|E|$ last. 
    \begin{enumerate} 
        \item The number of nodes in $G$ is
        $$
        |V| = \psi \cdot |E_O| + |V_I||V_O| = \Theta\left( \frac{r^3}{c^{29/3}} \cdot \frac{a^3r^2}{c^2} + a'^2 \cdot a^3r \right) = \Theta\left(  \frac{a^3r^5}{c^{35/3}} \right)
        $$
        and the number of paths in $\Pi$ is $|\Pi| = |\Pi_O| = \Theta\left(\frac{a^2r^4}{c^5}\right)$, by \cref{thm:modified_outer_graph}.
        % \item Property 2 follows directly from the final step of the construction of $G$. 
        \item 
        By Property 3 of \cref{lem:sourcewise_dp}, if path $\pi \in \Pi$ passes through $G_I^v$, then the subpath $\pi_{s, t}$ of $\pi$ induced on $G_I^v$ contains $\Theta(a'/c'^{3/2})$ edges that do not lie on any other path $\pi_{s', t'}$ where $(s', t') \neq (s, t)$ and $(s', t') \in P_I$. Then if path $\pi \in \Pi$ does not contain $\Theta(a'/c'^{3/2})$ edges that do not lie on any other path $\pi' \in \Pi$, it follows that $(s, t) = (s', t')$.

        % Suppose for the sake of contradiction that critical paths $\pi, \pi' \in \Pi$ share an edge in inner graph copy $G_I^v$. Let $\pi_{s, t}$ and $\pi_{s', t'}$ be the subpaths of $\pi$ and $\pi'$ respectively in $G_I^v$, where $(s, t), (s', t') \in P_I$ in $G_I^v$. Then by Property 3 of \cref{lem:sourcewise_dp}, $(s, t) = (s', t')$, or otherwise $\pi_{s, t}$ and $\pi_{s', t'}$ are edge-disjoint.

        Let $P_{e_s}$ be a subdivided path incident to $s$ in $G_I^v$, and let $P_{e_t}$ be a subdivided path incident to $t$ in $G_I^v$. Note that subdivided paths $P_{e_s}$ and $P_{e_t}$ must exist by the construction of critical path $\pi \in \Pi$. By Property 1 of \cref{claim:functions}, $P_{e_s}$ and $P_{e_t}$ are the unique subdivided paths incident to $s$ and $t$ respectively in $G_I^v$. Then this implies that $P_{e_s}$ and $P_{e_t}$  are both subpaths of critical paths $\pi$ and $\pi'$.

        Let $\pi_O = \phi^{-1}(\pi)$, and let $\pi'_O = \phi^{-1}(\pi')$. Then if $P_{e_s}, P_{e_t} \subseteq \pi \cap \pi'$, it follows that $\pi_O$ and $\pi_O'$ intersect on at least two edges, contradicting Property 3 of \cref{thm:modified_outer_graph}. We conclude that $\pi$ and $\pi'$ do not intersect on any inner graph edges in $E(G_I^v)$ for any $v \in V_O$. 
    \end{enumerate}
    We are now ready to prove the size bounds of $|E|$. The number of edges in $G$ is at most 
    $$
    |E| \leq \psi \cdot |E_O| + |E_I||V_O| = O\left(  \frac{r^3}{c^{29/3}} \cdot \frac{a^3r^2}{c^2} + a'^2c' \cdot a^3r \right) =  O\left(  \frac{a^3r^5}{c^{34/3}} \right) = O(c^{1/3} \cdot |V|).
    $$
    To obtain a lower bound on $|E|$, observe the following:
    \begin{itemize}
        \item Every critical path $\pi \in \Pi$ passes through at least $\frac{a}{4r} - 1$ distinct inner graph copies by \cref{thm:modified_outer_graph}.
        % \item If critical path $\pi \in \Pi$ passes through inner graph copy $G_I^v$, then $\pi$ contains a subpath $\pi_{s, t}$ in $G_I^v$, where $(s, t) \in P_I$, by construction. Then $|\pi_{s, t}| = \Theta\left(a'c'^{1/2}\right)$, by Property 4 of \cref{lem:sourcewise_dp}. 
        \item By Property 3 of this theorem, every path $\pi \in \Pi$ that passes through inner graph copy $G_I^v$ contains $\Theta(a'/c'^{3/2})$ edges that do not lie on any other path $\pi' \in \Pi$, for all $v \in V_O$.
    \end{itemize}
    Then combining these three observations we obtain:
    $$
    |E| \geq |\Pi| \cdot \left(\frac{a}{4r} - 1\right) \cdot \Theta\left(\frac{a'}{c'^{3/2}}\right) = \Omega\left(  \frac{a^2r^4}{c^5} \cdot \frac{a}{r} \cdot \frac{r^2}{c^{19/3}} \right) = \Omega\left( \frac{a^3r^5}{c^{34/3}}  \right) = \Omega(c^{1/4} \cdot |V|).
    $$
\end{proof}

\section{Spanner Lower Bound Analysis}
% \ghnoteinline{Section 7 is stable, please read and comment at your leisure.}
\label{sec:spanner_analysis}

Fix a critical path $\pi^* \in \Pi$ with endpoints $s, t$ in $G$. We may assume that $s, t \in S_I$ are source nodes in two distinct copies of inner graph $G_I$ (e.g., by truncating path $\pi^*$ so that this condition holds).  Let $\vec{w}_1^* = (x, 0, x^2) \in W_1$ and $\vec{w}_2^* = (0, y, y^2) \in W_2$ be the two vectors used to construct path $\pi^*$. Let $\pi$ be any alternate simple $s \leadsto t$ path. The majority of our analysis will be towards proving that if $\pi$ takes a subdivided edge not in $\pi^*$, then $\pi$ is much longer than $\pi^*$; specifically, we show that in this case, $|\pi| - |\pi^*| = \Omega(\psi)$ (see Lemma \ref{lem:obstacle}). Once we prove this lemma, the remainder of the proof will follow from standard arguments in prior work.  

We begin our analysis by importing some basic definitions and claims from \cite{BodwinH22}.

\begin{definition}[cf. \cite{BodwinH22}, Moves]
\label{def:moves}
Let $\pi$ be a $u \leadsto v$ path in $G$ from some source $u \in S_I$ in some inner graph copy $G_I^{(1)}$ to some source $v \in S_I$ in some inner graph copy $G_I^{(2)}$. If no internal node of $\pi$ is a source node in $S_I$ in some inner graph copy $G_I^{(3)}$, where $G_I^{(3)} \neq G_I^{(1)}$, then we call $\pi$ a \textit{move}.  We define the following categories of moves in $G$.
\begin{itemize}
    \item \textbf{Forward Move.} We say that a path $\pi$ is a forward move if it travels from $u \in S_I$ to some sink $w \in T_I$ in $G_I^{(1)}$, and then takes a subdivided path $P_e$ from $w$ to reach a source $v \in S_I$ in $G_I^{(2)}$.
    \item \textbf{Backward Move.} We say that a path $\pi$ is a backward move if it
    travels from $u \in S_I$ to some source node $u_1 \in S_I$ in $G_I^{(1)}$, and then takes a subdivided path $P_e$ incident to $u_1$ to reach some sink $w \in T_I$  in $G_I^{(2)}$, and then travels to a source $v \in S_I$ in $G_I^{(2)}$. 
    \item \textbf{Zigzag Move.} We say that a path $\pi$ is a zigzag move if it travels from $u \in S_I$ to some source node $u_1 \in S_I$ in $G_I^{(1)}$, and then takes a subdivided path $P_{e_1}$ incident to $u_1$ to reach a sink $w_1 \in T_I$ in some inner graph copy $G_I^{(3)}$, and then travels to a sink $w_2 \in T_I$ in $G_I^{(3)}$ and takes a subdivided path $P_{e_2}$ incident to $w_2$ to reach a source vertex $v \in S_I$ in $G_I^{(2)}$. 
\end{itemize}
\end{definition}

\begin{claim}[cf. \cite{BodwinH22}]
    Each simple $s \leadsto t$ path $\pi$ in $G$ can be decomposed into a sequence of pairwise internally vertex-disjoint moves. 
\end{claim}
\begin{proof}
 Let source node $s \in S_I$ be in inner graph copy $G_I^{(1)}$. Let $s_1 \in S_I$ be the first source node in $\pi$ that belongs to an inner graph copy $G_I^{(2)}$, where $G_I^{(1)} \neq G_I^{(2)}$. Then subpath $\pi[s, s_1]$ is either a forward, backward, or zigzag move, by construction. More generally, if $s_i \in S_I$ is in inner graph copy $G_I^{(i)}$, then  define $s_{i+1} \in S_I$ to be the first source node in subpath $\pi[s_i, t]$ that belongs to an inner graph copy $G_I^{(i+1)}$, where $G_I^{(i+1)} \neq G_I^{(i)}$. Then each subpath $\pi[s_i, s_{i+1}]$ is either a forward, backward, or zigzag move. Moreover, these subpaths will be pairwise internally vertex-disjoint since $\pi$ is simple.
\end{proof}

Note that critical path $\pi^*$ is an  $s \leadsto t$ path that decomposes into a sequence of forward moves. The forward moves in $\pi^*$ alternate between taking the  subdivided path corresponding to vector $\vec{w}_1 \in W_1$ and the subdivided path corresponding to vector $\vec{w}_2 \in W_2$. We will compare the length of the critical path $\pi^*$ with the length of the arbitrary path $\pi$ by comparing the moves in the move decompositions of the two paths. 

Borrowing notation from \cite{BodwinH22}, let $m_1, \dots, m_k$ be the move decomposition of a simple $s \leadsto t$ path $\pi$, where $m_i$ is a move from source node $s_i \in S_I$ in inner graph copy $G_I^{(i)}$ to source node $s_{i+1}$ in inner graph copy $G_I^{(i+1)}$, and $s_1 = s$ and $s_{k+1} = t$. If inner graph $G_I^{(i)}$ is the image of a vertex $v \in V_O$ in $G_O$ (i.e., $G_I^{(i)} = G_I^v$), then we let $\coord(G_I^{(i)})$ be the integer vector in $\mathbb{R}^2$ with the coordinates of $v$. 

\begin{definition}[cf. \cite{BodwinH22}, Move vector]
The move vector $\vec{m}_i$ corresponding to move $m_i$ is
$$
\vec{m}_i := \coord(G_I^{(i+1)}) - \coord(G_I^{(i)}).
$$    
\end{definition}

% Let $\vec{w}_1^* = (x, 0, x^2) \in W_1(r, c)$ and $\vec{w}_2^* = (0, y, y^2) \in W_2(r, c)$ be the two vectors used to construct path $\pi^*$. \ghnote{is this clear enough?}

\noindent
We  define the\textit{ direction vector }$\Vec{d}^*$ associated with $\pi^*$ as $\vec{d}^* = (2x, 2y, -1)$. We will use the projection of $\Vec{m}_i$ onto $\Vec{d}^*$ in order to measure how much progress move $m_i$ is making in travelling towards the end node $t$. 

\begin{definition}[cf. \cite{BodwinH22}, Move distance] The move distance $d_i$ of move vector $\vec{m}_i$ is defined to be 
$$
d_i := \proj_{\vec{d}^*}\vec{m}_i,
$$
where $\proj_{\vec{d}^*}\vec{m}_i$ denotes the scalar projection of $\vec{m}_i$ onto $\vec{d}^*$. 
\end{definition}

We now define a quantity $\mu$ called the unit length of $\pi^*$ that relates distances in $G$ to Euclidean distances. 

\begin{definition}[cf. \cite{BodwinH22}, Unit length of $\pi^*$] Let $\nu = \proj_{\Vec{d}^*}(\coord(G_I^{(k+1)}) - \coord(G_I^{(1)}))$. We define the unit length $\mu$ of $\pi^*$ as 
$$
\mu := \frac{|\pi^*|}{\nu}.
$$
\end{definition}

% Using the definitions we hav

We are ready to define the key quantity $\Delta(m_i)$ associated with each move $m_i$ in the move decomposition. The move length difference $\Delta(m_i)$ captures (in an amortized sense) the difference between the length $|m_i|$ of move $m_i$ and an equivalent move in $\pi^*$. 

\begin{definition}[cf. \cite{BodwinH22}, Move length difference] We define the move length difference of a move $m_i$ as
$$
\Delta(m_i) := |m_i| - \mu \cdot  d_i.
$$
\end{definition}

The following claim verifies our intuition about $\Delta(m_i)$.

\begin{claim}[cf. \cite{BodwinH22}]
    $\sum_i \Delta(m_i) = |\pi| - |\pi^*|$
    \label{claim:move_length_sum}
\end{claim}
\begin{proof}
    We restate this proof for completeness. We have:
    $$
    \sum_i \Delta(m_i) = \sum_i |m_i| - \mu \sum_i d_i = |\pi| - \frac{|\pi^*|}{\nu} \cdot \sum_i d_i = |\pi| - |\pi^*|.
    $$
    The final equality is due to the sequence of equalities:
    $$
    \sum_i d_i = \sum_i \proj_{\vec{d}^*}\vec{m}_i = \proj_{\vec{d}^*}(\coord(G_I^{(k+1)}) - \coord(G_I^{(1)})) = \nu,
    $$
    since $\pi$ being an $s \leadsto t$ path implies $\sum_i \Vec{m}_i = \coord(G_I^{(k+1)}) - \coord(G_I^{(1)})$.
\end{proof}

Now what remains is to lower bound $\Delta(m_i)$ for each of the moves $m_i$ in our move decomposition. The following technical claim will be useful in that respect. 

\begin{claim}
\label{claim:unit_length}
    $$
       \mu =   (2\psi + \Theta(a'c'^{1/2})) \cdot \frac{\|\Vec{d}^*\|}{x^2 + y^2}
    $$
\end{claim}
\begin{proof}
Let $m_1, m_2$ be two consecutive moves in $\pi^*$.
Then
$$
m_1 = \pi_{s_1, t_1} \circ P_{e_1} \text{\qquad and \qquad} m_2 = \pi_{s_2, t_2} \circ P_{e_2},
$$
for some critical pairs $(s_1, t_1), (s_2, t_2) \in P_I$ and some subdivided paths $P_{e_1}$ and $P_{e_2}$. 
We may assume wlog that $\Vec{m}_1 = \Vec{w}_1^*$ and $\Vec{m}_2 = \Vec{w}_2^*$.  Then using the fact that 
$\coord(G_I^{(k+1)}) - \coord(G_I^{(1)}) = i(\Vec{w}_1^* + \Vec{w}_2^*)$ for some positive integer $i$, 
we calculate as follows:
\begin{align*}
    \mu & = \frac{|\pi^*|}{\proj_{\Vec{d}^*}(\coord(G_I^{(k+1)}) - \coord(G_I^{(1)}))} \\
    & = \frac{|m_1| + |m_2|}{\proj_{\Vec{d}^*}(\Vec{m}_1 + \Vec{m}_2)} \\
    & = \frac{|\pi_{s_1, t_1}| + |P_{e_1}| + |\pi_{s_2, t_2}| + |P_{e_2}|}{\proj_{\Vec{d}^*}(\Vec{m}_1 + \Vec{m}_2)} \\
    & = \frac{(\psi + \Theta(a'c'^{1/2})) + (\psi + \Theta(a'c'^{1/2}))}{\proj_{\Vec{d}^*}(\Vec{m}_1 + \Vec{m}_2)} & \text{Property 4 of \cref{lem:sourcewise_dp}} \\
    & = \frac{2\psi + \Theta(a'c'^{1/2}) }{\proj_{\Vec{d}^*}(\Vec{m}_1 + \Vec{m}_2)} & \text{\cref{claim:move_len}} \\
    & = \frac{(2\psi + \Theta(a'c'^{1/2}))\|\Vec{d}^*\|}{\langle \Vec{w}_1^* + \Vec{w}_2^*, \Vec{d}^*\rangle}  \\
    & = \frac{(2\psi + \Theta(a'c'^{1/2})) \|\Vec{d}^*\|}{x^2 + y^2}.
\end{align*}
\end{proof}

We will now lower bound the move length difference $\Delta(m_i)$ of backward moves $m_i$. 

\begin{claim}
    Let $m_i$ be a backward move. Then $\Delta(m_i)  = \Omega(\psi)$.
    \label{claim:backward}
\end{claim}
\begin{proof}
    Note that if $m_i$ is a backward move, then it takes a subdivided path $P_e$ that corresponds to a vector $-\Vec{w} \in \mathbb{R}^2$, where  in $\Vec{w} \in W_1 \cup W_2$. We may assume wlog that $\Vec{w} \in W_1$, as the case where $\Vec{w} \in W_2$ is symmetric. Let $\Vec{w} = (w, 0, w^2) \in W_1$. Then we obtain the following upper bound on $d_i$:
    $$
    d_i = \proj_{\vec{d}^*}\vec{m}_i = \frac{\langle \Vec{m}_i, \Vec{d}^* \rangle}{\|\Vec{d}^*\|} =  \frac{\langle -\Vec{w}, \Vec{d}^* \rangle}{\|\Vec{d}^*\|} =  \frac{-2xw + w^2}{\|\Vec{d}^*\|} = \frac{(w - 2x)w}{\|\Vec{d}^*\|} \leq 0,
    $$
    where the final inequality follows from the fact that $x, w \in [r/2, r]$ and so $w \leq 2x$ due to Claim \ref{claim:intervals}. Then
    $$
    \Delta(m_i) = |m_i| \geq \psi,
    $$
    since $m_i$ contains a subdivided path $P_e$ of length $|P_e| = \psi$. 
\end{proof}

We will now lower bound the move length difference $\Delta(m_i)$ of zigzag moves $m_i$. 

\begin{claim}
    Let $m_i$ be a zigzag move. Then $\Delta(m_i) = \Omega(\psi)$.
    \label{claim:zigzag}
\end{claim}
\begin{proof}
Let $G_I^{(1)},  G_I^{(2)}, G_I^{(3)}$ and  $P_{e_1}, P_{e_2}$ be the inner graph copies and subdivided paths, respectively, associated with zigzag move $m_i$, as defined in \cref{def:moves}. Note that subdivided path $P_{e_1}$ corresponds to a vector $-\Vec{w}_1 \in \mathbb{R}^2$, where $\Vec{w}_1$  in $\Vec{w}_1 \in W_1 \cup W_2$. Then by an argument identical to that of \cref{claim:backward}, 
$$
\proj_{\vec{d}^*} (\coord( G_I^{(2)}) - \coord( G_I^{(1)})) = \proj_{\vec{d}^*} (-\Vec{w}_1)  \leq 0.
$$
On the other hand, $P_{e_2}$ corresponds to a vector $\vec{w}_2  \in W_1 \cup W_2$. We may assume wlog that $\vec{w}_2 = (w, 0, w^2) \in W_1$, as the case where $\Vec{w} \in W_2$ is symmetric. Then
$$
\proj_{\vec{d}^*} (\coord( G_I^{(3)}) - \coord( G_I^{(2)})) = \proj_{\vec{d}^*} \Vec{w}_2  = \frac{2xw - w^2}{\|\Vec{d}^*\|}.
$$
Taking the above two inequalities together we get
$$
d_i = \proj_{\vec{d}^*}\vec{m}_i = \proj_{\vec{d}^*} (\coord( G_I^{(2)}) - \coord( G_I^{(1)})) + \proj_{\vec{d}^*} (\coord( G_I^{(3)}) - \coord( G_I^{(2)})) \geq  \frac{2xw - w^2}{\|\Vec{d}^*\|}.
$$
Now note the following brief observations:
\begin{itemize}
    \item $d_i \geq 0$, since $2xw - w^2 \geq 0$, as $w \leq 2x$ by \cref{claim:stripes},
    \item  $|x -y| \leq  r/(16c^3)$ by \cref{claim:stripes}, since $\Vec{w}_1^* = (x, 0, x^2)$ and $\Vec{w}_2^* = (0, y, y^2)$ satisfy $\Vec{w}_1^*, \Vec{w}_2^* \in \mathcal{S}_i^1 \cup \mathcal{S}_i^2$ for some $i \in [1, c]$, and
    \item $r \leq 2x$, since $r/2 \leq x$ by \cref{claim:stripes}.  
\end{itemize}
 
% Additionally, note that 
% \ghnoteinline{gary barrier}
Then we can lower bound $\Delta(m_i)$ as follows:
\begin{align*}
    \Delta(m_i) & = |m_i| - \mu \cdot d_i \\
    & \geq |P_{e_1}| + |P_{e_2}| - \mu \cdot d_i \\
    & \geq 2\psi -  (2\psi + \Theta(a'c'^{1/2})) \cdot \frac{\|\Vec{d}^*\|}{x^2 + y^2} \cdot \frac{2xw - w^2}{\|\Vec{d}^*\|} & \text{ by \cref{claim:unit_length}} \\
    & \geq 2\psi -  (2\psi + \Theta(a'c'^{1/2})) \cdot \frac{2xw - w^2}{x^2 + y^2}  \\
    & \geq 2\psi -  (2\psi + \Theta(a'c'^{1/2})) \cdot \frac{2xw - w^2}{x^2 + (x - r/(16c^3))^2} & \text{ by \cref{claim:stripes}}  \\
    & \geq 2\psi -  (2\psi + \Theta(a'c'^{1/2})) \cdot \frac{2xw - w^2}{2x^2 -xr/c}   \\
    & \geq 2\psi -  (2\psi + \Theta(a'c'^{1/2})) \cdot \frac{2xw - w^2}{(1 - 1/c) \cdot 2x^2}. & \text{ by \cref{claim:stripes}}
\end{align*}
Note that $2xw - w^2$ is maximized when $w = x$, so $2xw - w^2 \leq x^2$. Then we conclude that
\begin{align*}
     \Delta(m_i) & \geq 2\psi -  (2\psi + \Theta(a'c'^{1/2})) \cdot \frac{x^2}{(1 - 1/c) \cdot 2x^2} \\
     & \geq 2\psi - (2\psi + \Theta(a'c'^{1/2})) \cdot \frac{1}{(3/4) \cdot 2} & \text{ by taking $c \geq 4$} \\
          & = 2\psi - (2\psi + \Theta(a'c'^{1/2})) \cdot \frac{2}{3}  \\
          & \geq \frac{2\psi}{3} - \Theta(a'c'^{1/2}) = \Omega(\psi) & \text{since $a'c'^{1/2} = o(\psi)$.}
\end{align*}
Recall that $\psi = \Theta_c\left( r^3 \right)$ and $a'c'^{1/2} = \Theta_c(r^2)$.  Since input parameter $r$ will grow with the input size while $c$ will be taken to be a sufficiently large constant, the final equality follows. 
\end{proof}

% \begin{claim}
%     Let $m_i$ be a stationary move. Then $\Delta(m_i) \geq 0$.
% \end{claim}
% \begin{proof}
    
% \end{proof}

All subdivided paths $P_e$ in $\pi^*$ correspond to vectors $\vec{w}_1^* = (x, 0, x^2) \in W_1$ and $\Vec{w}_2^* = (0, y, y^2) \in W_2$. Moreover, by construction there exists some index $i \in [1, c]$ such that $\Vec{w}_1^* \in \mathcal{S}_i^1$ and $\Vec{w}_2^* \in \mathcal{S}_i^2$. Denote this index as $i^*$. 
 Now consider a forward move $m$. Move $m$ begins at source node $u \in S_I^i$ in inner graph copy $G_I^{(1)}$ and contains a subdivided path $P_e$ as a suffix. This subdivided path $P_e$ is incident to a unique sink $w \in T_I^j$ in $G_I^{(1)}$.

To analyze forward moves in our move decomposition of $\pi$, it will be useful to partition them into three different sets:
\begin{itemize}
    \item \textbf{Forward Move with Different Sink Stripe.} We say that a move $m$ is a forward move with different sink stripe than $\pi^*$ if $m$ is a forward move from source node $u \in S_I^i$ to sink node $w \in T_I^j$, where $j \neq i^*$. 
    \item \textbf{Forward Move with Different Source Stripe.} We say that a move $m$ is a forward move with different source stripe than $\pi^*$ if $m$ is a forward move from source node $u \in S_I^i$ to sink node $w \in T_I^j$, where $i \neq i^*$ and $j = i^*$. 
    \item \textbf{Forward Move with Same Stripes.} We say that a move $m$ is a forward move with the same  stripes as $\pi^*$ if $m$ is a forward move from source node $u \in S_I^i$ to sink node $w \in T_I^j$, where $i = j = i^*$. 
\end{itemize}
% \ghnoteinline{Was the above definitions clear enough?}

\begin{claim}
    Let $m_i$ be a forward move with a different sink stripe than $\pi^*$. Then $\Delta(m_i) \geq \Omega(\psi/c^2)$.
    \label{claim:sink}
\end{claim}
\begin{proof}
    Since $m_i$ is a forward move with a different sink stripe than $\pi^*$, it follows that the subdivided path $P_e$ taken by $m_i$ corresponds to a vector $\Vec{w}$ in $\mathcal{S}_j^1 \cup \mathcal{S}_j^2$, where $j \neq i^*$. We may assume wlog that $\Vec{w} = (w, 0, w^2) \in \mathcal{S}_j^1 \subseteq W_1$, as the case where $\Vec{w} \in W_2$ is symmetric. Moreover, by  \cref{claim:stripes}, $|w - x| \geq r/(2c)$ and $x \leq r$. Then we obtain the following inequality for $d_i$:
    $$
    d_i = \proj_{\Vec{d}^*}\Vec{m}_i = \frac{\langle \Vec{w}, \Vec{d}^* \rangle}{\|\Vec{d}^*\|} = \frac{ 2xw - w^2 }{\|\Vec{d}^*\|} = \frac{ x^2 - (w-x)^2 }{\|\Vec{d}^*\|} \leq \frac{ x^2 - r^2/(4c^2) }{\|\Vec{d}^*\|} \leq \left(1- \frac{1}{4c^2}\right)  \cdot \frac{  x^2 }{\|\Vec{d}^*\|}.
    $$
    By \cref{claim:intervals}, since $\Vec{w}_1^*, \Vec{w}_2^* \in \mathcal{S}_{i^*}^1 \cup \mathcal{S}_{i^*}^2$, it follows that $|x - y| \leq \frac{r}{16c^3}$, so $$x^2 + y^2 \geq x^2 + \left(x-  \frac{r}{16c^3}  \right)^2 \geq 2x^2 - \frac{2xr}{16c^3} \geq \left( 1 - \frac{1}{8c^3}\right) \cdot 2x^2.$$ 
    Combining  our two inequalities, we can now bound $\Delta(m_i)$ as follows:
    \begin{align*}
        \Delta(m_i) & \geq |m_i| - \mu \cdot d_i \\
        & \geq |P_e| - (2 \psi+\Theta(a'c'^{1/2})) \cdot \frac{\|\Vec{d}^*\|}{x^2 + y^2} \cdot  \left(1- \frac{1}{4c^2}\right)  \cdot \frac{  x^2 }{\|\Vec{d}^*\|} \\
        & \geq \psi - (2 \psi+\Theta(a'c'^{1/2})) \cdot \frac{x^2}{x^2 + y^2} \cdot  \left(1- \frac{1}{4c^2}\right) \\
        & \geq \psi - (2 \psi+\Theta(a'c'^{1/2})) \cdot \frac{x^2}{\left( 1 - \frac{1}{8c^3}\right) \cdot 2x^2} \cdot  \left(1- \frac{1}{4c^2}\right) \\
        & \geq \psi - (\psi+\Theta(a'c'^{1/2})) \cdot \frac{\left(1- \frac{1}{4c^2}\right)}{\left( 1 - \frac{1}{8c^3}\right) }  \\
        & = \psi - (\psi+\Theta(a'c'^{1/2})) \cdot \frac{8c^3 \cdot (4c^2-1)}{4c^2 \cdot (8c^3 - 1)}  \\
        & = \psi - (\psi+\Theta(a'c'^{1/2})) \cdot \frac{8c^3 - 2c}{8c^3 - 1}  \\
        & \geq \psi - (\psi+\Theta(a'c'^{1/2})) \cdot \left( 1 - \frac{1}{4c^2}  \right) \\
        & \geq \frac{1}{4c^2} \cdot \psi - \Theta(a'c'^{1/2})  \\ & = \Omega(\psi/c^2).
    \end{align*}
    The final equality holds since $ \psi = \Theta_c\left( r^3 \right)$ and $a'c'^{1/2} = \Theta_c(r^2)$, and input parameter $r$ will grow with the input size while $c$ will be taken to be a sufficiently large constant.
\end{proof}

Note that since the critical path $\pi^*$ has associated vectors $\Vec{w}_1^*, \vec{w}_2^* \in \mathcal{S}_{i^*}^1 \cup \mathcal{S}_{i^*}^2$, it follows that path $\pi^*$ contains a path $\pi_{u, v}$  between critical pair $(u, v) \in S_I^{i^*} \times T_I^{i^*}$. By Property 5 of Lemma \ref{lem:sourcewise_dp}, $\dist_{G_I}(u, v) = \dist_{G_I}(S_I^{i^*}, T_I^{i^*})$. Let $\lambda = \dist_{G_I}(S_I^{i^*}, T_I^{i^*})$. Then we have the following two claims. 
\begin{claim}
    Let $m$ be a  move in $\pi^*$. Then $|m| = \lambda + \psi$.
    \label{claim:move_len}
\end{claim}
\begin{proof}
Each move $m$ in $\pi^*$ is of the form $m = \pi_{s, t} \circ P_e$, where $\pi_{s, t}$ is an $s \leadsto t$ path such that $(s, t) \in S_I^{i^*} \times T_I^{i^*}$ of length $|\pi_{s, t}| = \lambda$ and $P_e$ is a subdivided edge of length $|P_e| = \psi$. 
\end{proof}
\begin{claim}
    \label{claim:unit_distance_exact}
    $$\mu = \frac{2(\lambda + \psi)\|\Vec{d}^*\|}{x^2 + y^2}$$
\end{claim}
\begin{proof}
Let $m_1, m_2$ be two consecutive moves in $\pi^*$. We may assume wlog that $\Vec{m}_1 = \Vec{w}_1^*$ and $\Vec{m}_2 = \Vec{w}_2^*$.  Then using the fact that 
$\coord(G_I^{(k+1)}) - \coord(G_I^{(1)}) = i(\Vec{w}_1^* + \Vec{w}_2^*)$ for some positive integer $i$, 
we calculate as follows:
\begin{align*}
    \mu & = \frac{|\pi^*|}{\proj_{\Vec{d}^*}(\coord(G_I^{(k+1)}) - \coord(G_I^{(1)}))} \\
    & = \frac{|m_1| + |m_2|}{\proj_{\Vec{d}^*}(\Vec{m}_1 + \Vec{m}_2)} \\
    & = \frac{2(\lambda + \psi) }{\proj_{\Vec{d}^*}(\Vec{m}_1 + \Vec{m}_2)} & \text{\cref{claim:move_len}} \\
    & = \frac{2(\lambda + \psi) \|\Vec{d}^*\|}{\langle \Vec{w}_1^* + \Vec{w}_2^*, \Vec{d}^*\rangle}  \\
    & = \frac{2(\lambda + \psi)\|\Vec{d}^*\|}{x^2 + y^2}.
\end{align*}
\end{proof}

We can now bound the move length difference of forward moves with different source stripes than $\pi^*$. These moves can actually have negative move length difference, which will pose difficulties in our analysis later.

\begin{claim}
    Let $m_i$ be a forward move with a different source stripe than $\pi^*$. Then $\Delta(m_i) \geq - \psi/c^3$. 
    \label{claim:source}
\end{claim}
\begin{proof}
Move $m_i$ containing a subdivided path $P_e$, so $m_i \geq |P_e| = \psi$.  The subdivided path $P_e$ taken by $m_i$ corresponds to a vector $\Vec{w}$ in $W_1 \cup W_2$. We may assume wlog that $\Vec{w} = (w, 0, w^2) \in W_1$, as the case where $\Vec{w} \in W_2$ is symmetric. Then  we have the following bound on $\Delta(m_i)$:
\begin{align*}
    \Delta(m_i) & \geq |m_i| - \mu \cdot d_i \\
    & \geq \psi -\frac{ 2(\lambda + \psi) \cdot \|\Vec{d}^*\|}{x^2 + y^2} \cdot \frac{2wx - w^2}{\|\Vec{d}^*\|}  & \text{\cref{claim:unit_distance_exact}} \\
    & \geq \psi -  \frac{2(\lambda + \psi) \cdot x^2}{x^2 + y^2}& \text{since $2wx - w^2$ is maximized when $w =x$} \\
        & \geq \psi - \frac{ 2(\lambda + \psi) \cdot x^2}{\left( 1 - \frac{1}{8c^3} \right) \cdot 2x^2} & \text{\cref{claim:stripes}} \\
        & \geq \psi - (\lambda+ \psi) \cdot \frac{8c^3}{8c^3 - 1} \\
        & \geq -\frac{ \psi}{4c^3} - 2\lambda & \text{sufficiently large $c$} \\
        & \geq - \psi/c^3,
\end{align*}
where the final inequality follows from the fact that $ \psi/c^3 = \Theta_c\left( r^3 \right)$ and $\lambda = \Theta(a'c'^{1/2}) = \Theta_c(r^2)$, and that input parameter $r$ will grow with the input size of our graph construction. 
\end{proof}

Note that while a forward move $m_i$ with a different source stripe than $\pi^*$ can have negative move length difference, this move is always preceded by a move $m_{i-1}$ with a large positive move length difference. This will allow us to `charge' the negative length of forward moves with different source stripes to the moves preceding them in an amortized argument in \cref{claim:obstacle_1}. We now verify our desired claim about forward moves with different source stripes.

\begin{claim}
Every forward move $m_i$ with a different source stripe than $\pi^*$ is immediately preceded by a move $m_{i-1}$, where $m_{i-1}$ is either a backward move, a zigzag move, or a forward move with different sink stripe than $\pi^*$. 
\label{claim:preceded}
\end{claim}
\begin{proof}
    Note that by definition, the start node of move $m_i$ is a source node $u$ in $S_I^i$, where $i \neq i^*$. If move $m_{i-1}$ is a forward move, then it must take a subdivided path $P_e$ that ends at source node $u$. Since $u \not \in S_I^{i^*}$, Property 1 of  \cref{claim:functions} implies that if move $m_{i-1}$ is a forward move, then it must be a forward move with different sink stripe than $\pi^*$. We conclude that move $m_{i-1}$ is either a backward move, a zigzag move, or a forward move with different sink stripe than $\pi^*$. 
\end{proof}

\begin{claim}
\label{claim:same}
    Let $m_i$ be a forward move with the same source and sink stripe as $\pi^*$. Let $P_e$ be the subdivided path used by move $m_i$, and let $\Vec{w} \in W_1 \cup W_2$ be the vector corresponding to $P_e$. Then we have the following bound on $\Delta(m_i)$:
    \begin{itemize}
        \item if $\Vec{w} \in W_1$, then $\Delta(m_i)  \geq (\lambda + \psi) -  \frac{2(\lambda + \psi)x^2}{x^2+y^2}$, and
        \item if $\Vec{w} \in W_2$, then $\Delta(m_i) \geq  (\lambda + \psi) - \frac{2(\lambda + \psi)y^2}{x^2+y^2}$.
    \end{itemize}
    Moreover, this implies that $\Delta(m_i) \geq -\psi/c^3$. 
\end{claim}
\begin{proof}
Let $P_e$ be the subdivided path used by move $m_i$. Let $s_i \in S_I^{i^*}$ be the first source node in $m_i$ and let $t_i \in T_I^{i^*}$ be the last sink node in $m_i$.  By Lemma Property 5 of \ref{lem:sourcewise_dp}, $$\dist_{G_I}(s_i, t_i) = \dist_{G_I}(S_I^{i^*}, T_I^{i^*}) = \lambda.$$ It follows that $|m_i| = |\pi_{s_i, t_i}| + |P_e| = \lambda + \psi$. We now split our analysis into two cases:
\begin{itemize}
    \item  $\Vec{w} \in W_1$. Let $\Vec{w} = (w, 0, w^2)$. Then $\langle \Vec{w}, \Vec{d}^* \rangle = 2wx - w^2 \leq x^2$, by the same argument as in \cref{claim:source}. 
    It follows that by \cref{claim:unit_distance_exact}, 
    $$\Delta(m_i)  = |m_i| - \mu \cdot d_i = (\lambda+\psi) - \frac{2(\lambda+\psi) \cdot \|\Vec{d}^*\|}{x^2 + y^2} \cdot \frac{\langle \Vec{w}, \Vec{d}^* \rangle}{\|\Vec{d}^*\|} \geq (\lambda+\psi) -  \frac{2(\lambda+\psi)x^2}{x^2+y^2}.$$
    \item  $\Vec{w} \in W_2$. Let $\Vec{w} = (0, w, w^2)$. Then $\langle \Vec{w}, \Vec{d}^* \rangle = 2wy - w^2 \leq y^2$ by the same argument as in \cref{claim:source}.  It follows that by \cref{claim:unit_distance_exact}, 
    $$\Delta(m_i)  = |m_i| - \mu \cdot d_i = (\lambda+\psi) -  \frac{2(\lambda+\psi)\cdot\|\Vec{d}^*\|}{x^2 + y^2} \cdot \frac{\langle \Vec{w}, \Vec{d}^* \rangle}{\|\Vec{d}^*\|} \geq (\lambda+\psi) -  \frac{2(\lambda+\psi)y^2}{x^2+y^2}.$$
\end{itemize}
The final claim follows from the inequalities $x^2 + y^2 \geq \left( 1 - \frac{1}{8c^3}\right) \cdot 2x^2$ and  $x^2 + y^2 \geq \left( 1 - \frac{1}{8c^3}\right) \cdot 2y^2$ proved in \cref{claim:sink} and following from \cref{claim:stripes}. Formally,
\begin{align*}
   \Delta(m_i)  & \geq \min \left( (\lambda+\psi) -  \frac{2(\lambda+\psi)x^2}{x^2+y^2}, \quad 
(\lambda+\psi) -  \frac{2(\lambda+\psi)y^2}{x^2+y^2}
\right)\\
& \geq
\min \left( (\lambda+\psi) -  \frac{2(\lambda+\psi)x^2}{\left( 1 - \frac{1}{8c^3}\right) \cdot 2x^2}, \quad 
(\lambda+\psi) -  \frac{2(\lambda+\psi)y^2}{\left( 1 - \frac{1}{8c^3}\right) \cdot 2y^2}
\right)   \\
& \geq
\min \left( (\lambda+\psi) -   \frac{8c^3(\lambda+\psi)}{8c^3 - 1}, \quad 
(\lambda+\psi) -  \frac{8c^3(\lambda+\psi)}{8c^3 - 1}
\right) \\
& \geq -\psi/c^3. & \text{since $\lambda = o(\psi/c^3)$} 
\end{align*}
\end{proof}

We are now ready to prove our amortized argument about the sums of move lengths $\sum_i \Delta(m_i)$. Our goal will be to prove a lower bound for $\sum_i \Delta(m_i)$ and then use \cref{claim:move_length_sum} to obtain a lower bound for $|\pi| - |\pi^*|$.  

\begin{claim}
    Let $\pi$ be an $s \leadsto t$ path in $G$ that contains a backwards move,  a zigzag move, or a forward move with different sink stripe than $\pi^*$. Then $|\pi| - |\pi^*| = \Omega(\psi/c^2)$. 
    \label{claim:obstacle_1}
\end{claim}
\begin{proof}
Let $m_1, \dots, m_{k}$ be the move decomposition of $\pi$. Recall that $\sum_i \Delta(m_i) = |\pi| - |\pi^*|$ by \cref{claim:move_length_sum}.  In order to prove this claim, we will need to introduce a simple charging scheme to our analysis. We define the following operation on pairs of moves $m_i, m_j$ in the move decomposition:
\begin{align*}
    \Delta(m_i) &:= \Delta(m_i) + \Delta(m_j) \\
    \Delta(m_j) &:= 0.
\end{align*}
If this operation is performed, we say that we charged the cost of move $m_j$ to move $m_i$. Note that the sum $\sum_i \Delta(m_i)$ is invariant under this operation.

By our earlier analysis, if $\Delta(m_i) < 0$, then move move $m_i$ is either a forward move with a different source stripe than $\pi^*$ or a forward move with the same source and sink stripes as $\pi^*$. We handle these two cases separately as follows:
\begin{itemize}
    \item If move $m_i$ is  a forward move with a different source stripe than $\pi^*$, then by Claim \ref{claim:preceded}, move $m_{i-1}$ is either a backward move, a zigzag move, or a forward move with different sink stripe than $\pi^*$. In any of these cases $\Delta(m_{i-1}) \geq  \Omega(\psi/c^2)$ by Claims \ref{claim:backward}, \ref{claim:zigzag}, and \ref{claim:sink}. Then since $\Delta(m_i) \geq - \psi/c^3$ by \cref{claim:source}, we can safely charge the cost of move $m_i$ to move $m_{i-1}$ by taking $c$ to be a sufficiently large constant.
    \item If move $m_i$ is a forward move with the same source and sink stripes as $\pi^*$, then we will again split our analysis as follows: 
    \begin{itemize}
        \item If move $m_{i-1}$ or move $m_{i-2}$ is a backward move, a zigzag move, or a forward move with different sink stripe than $\pi^*$, then we will charge the cost of move $m_i$ to move $m_{i-1}$ or move $m_{i-2}$, respectively. 
        \item If move $m_{i-1}$ is a forward move with the same source and sink stripe as $\pi^*$, then since forward moves alternate between taking subdivided paths in $W_1$ and $W_2$, it follows that by \cref{claim:same},
        $$
        \Delta(m_{i-1}) + \Delta(m_i) \geq \left( (\lambda+\psi) -  \frac{2(\lambda+\psi)x^2}{x^2+y^2} \right) + \left( (\lambda+\psi) -  \frac{2(\lambda+\psi)y^2}{x^2+y^2} \right) = 0.
        $$
        We charge $\Delta(m_i)$ to $\Delta(m_{i-1})$. After this operation, $\Delta(m_{i-1}) \geq \Delta(m_i) \geq 0$. 
        \item If move $m_{i-1}$ is  a forward move with a different source stripe than $\pi^*$, then by Claim \ref{claim:preceded}, move $m_{i-2}$ is either a backward move, a zigzag move, or a forward move with different sink stripe than $\pi^*$, so move $m_i$ is handled by the first case. 
    \end{itemize}
\end{itemize}
Note that the above  cases cover all possibilities for a move $m_i$, $i \not \in \{1, 2\}$. To handle moves $m_i$, $i \in \{1, 2\}$, we will charge moves $m_{i+1}$ and $m_{i+2}$ in an analogous way as above. 

After completing our charging operations, every forward move $m_i$ with the same source and sink stripes as $\pi^*$ now satisfies $\Delta(m_i) \geq 0$, so we may ignore all of these moves. Likewise, every forward move $m_i$ with a different source stripe than $\pi^*$ satisfies $\Delta(m_i) = 0$, so we may ignore all of these moves.

Note that all remaining moves are either  backwards moves, zigzag moves, or  forward moves with different sink stripes than $\pi^*$. 
Every time a remaining move $m_i$ is charged by a move $m_j$, we have that $\Delta(m_i) = \Omega(\psi/c^2)$ and $\Delta(m_j) \geq - \psi/c^3$. Moreover, each move is charged by at most four other moves (the two moves immediately preceding $m_i$ and the two moves immediately following $m_i$). Therefore, after all charging operations are finished, move $m_i$ satisfies $\Delta(m_i) \geq \Omega(\psi/c^2)$. 

Then after we have performed all our charging operations, every move $m_i$ in the move decomposition satisfies $\Delta(m_i) \geq 0$. 
Moreover, if the move decomposition of $\pi$ contains a move $m_i$ that is a backwards move, or a zigzag move, or a forward move with different sink stripe than $\pi^*$,  then after our charging operations, $\Delta(m_i) = \Omega(\psi/c^2)$. We conclude that $|\pi|-|\pi^*| = \sum_i \Delta(m_i) = \Omega(\psi/c^2)$. 
\end{proof}

\begin{claim}
\label{claim:obstacle_2}
Let $\pi$ be an $s \leadsto t$ path in $G$ such that 
\begin{itemize}
    \item $\pi$ does not contain a backwards move, nor a zigzag move, nor a forward move with different sink stripe than $\pi^*$, and
    \item $\pi$ takes a subdivided path $P_e$ not in $\pi^*$. 
\end{itemize}
Then $|\pi| - |\pi^*| = \Omega(\psi)$. 
\end{claim}
\begin{proof}
    Note that since $\pi$ does not contain a backwards move, nor a zigzag move, nor a forward move with different sink stripe than $\pi^*$, by \cref{claim:preceded} it follows that $\pi$ does not contain  a forward move with a different source stripe than $\pi^*$ either. Then $\pi$ contains exclusively forward moves with the same source and sink stripes as $\pi^*$. Then each move $m_i$ of $\pi$ is of length $|m_i| = \dist_{G_I}(S_I^{i^*}, T_I^{i^*}) + |P_e| = \lambda + \psi$, where $P_e$ is the subdivided path taken by move $m_i$. 
    
    Let $m_1, \dots, m_k$ be the move decomposition of $\pi$ and let $m_1, \dots, m_{k^*}$ be the move decomposition of $\pi^*$. Then since each move $m$ of $\pi$ is of length $|m| = \lambda + \psi$, it follows that $|\pi| = k(\lambda+ \psi)$. Likewise, by \cref{claim:unit_distance_exact}, $|\pi^*| = k^* \cdot (\lambda+ \psi)$. 
    
    Now observe that $\pi$ and $\pi^*$ are the images of paths $
    \phi^{-1}(\pi) = \pi_O$ and $\phi^{-1}(\pi^*) = \pi_O^*$, respectively, in outer graph $G_O$. Moreover, $|\pi_O| = q$ and $|\pi_O^*| = q^*$. 
    Recall that $\pi_O^*$ is a critical path between endpoints $(s_O, t_O) \in P_O$ in $G_O$. Additionally, $\pi_O$ is a $s_O \leadsto t_O$ path  in $G_O$ that contains an edge  not in $\pi_O^*$. Then since $\pi_O^*$ is a unique shortest path in $G_O$ by Lemma \ref{thm:modified_outer_graph}, it follows that $|\pi_O| \geq |\pi_O^*| + 1$. We conclude that $|\pi| \geq |\pi^*| + (\lambda+ \psi)$, as desired.  
\end{proof}

Claims \ref{claim:obstacle_1} and \ref{claim:obstacle_2} immediately imply the following lemma. 

\begin{lemma}
    Let $\pi$ be an $s \leadsto t$ path in $G$ that takes a subdivided path not in $\pi^*$. Then $|\pi| - |\pi^*| = \Omega(\psi/c^2)$. 
    \label{lem:obstacle}
\end{lemma}

We are ready to complete the proof of our spanner lower bounds.

\begin{theorem}
    For any sufficiently large parameter $c$, there are infinitely many $n$ for which there is an $n$-node graph $G$ such that any spanner of $G$ with at most $cn$ edges has additive distortion $+\Omega_c(n^{3/17})$. 
\end{theorem}
\begin{proof}
    Let $c > 0$ be a sufficiently large constant. Then we will construct the infinite family of obstacle product graphs $G$ from \cref{thm:spanner_obstacle_graph} with parameters $a, r, c_1$, where $c_1 = \Theta(c^3)$. Specifically, we choose $c_1$ to be large enough so that every graph $G$ on $n$ nodes and $m$ edges in our family satisfies $\frac{cn}{m} < \frac{1}{2}$. This is possible to achieve by setting $c_1 = \Theta(c^3)$ due to Property 1 of \cref{thm:spanner_obstacle_graph}. 

    Let $H$ be a spanner of $G$ with at most $cn$ edges. Then $H$ contains at most half the edges in $G$. By Lemma \ref{lem:obstacle} and Property 2 of \cref{thm:spanner_obstacle_graph}, if $H$ is missing an edge in a subdivided path $P_e$ in $G$, then $H$ has additive error $+\Omega(\psi/c^2)$. 

    Otherwise, if $H$ contains all edges in the subdivided paths in $G$, then at least half of the inner graph edges $E(G_I^v)$, $v \in V_O$, must be missing from $H$. Let $E_I := \cup_{v \in V_O} E(G_I^v)$. Since the paths in $\Pi$ partition the edges in $E_I$ by Property 3 of \cref{thm:spanner_obstacle_graph}, it follows that there must exist a critical path $\pi \in \Pi$ missing at least half of its edges in $E_I$ in $H$. 

    Let $s, t$ be the endpoints of a critical path $\pi \in \Pi$ in $G$, and let $\pi'$ be a simple $s \leadsto t$ path in $H$. If $\pi'$ takes a subdivided path not in $\pi$, then $\pi'$ suffers $+\Omega(\psi/c^2)$ additive error by Lemma \ref{lem:obstacle}. Then we may assume that $\pi'$ takes the exact same sequence of subdivided paths as $\pi$. 

    If path $\pi$ is missing an edge in inner graph copy $G_I^v$ in $H$, then path $\pi'$ must suffer at least $+1$ additive error in $H$ while passing through this inner graph, since $\pi[V(G_I^v)]$ is a unique shortest path in $G_I$ by Property 2 of \cref{lem:sourcewise_dp} and the construction of critical path $\pi$. Since $\pi$ is missing half of its edges in $E_I$ in $H$, path $\pi$ is missing an edge in at least half of the inner graph copies that it passes through. As path $\pi$ passes through $\Theta\left(\frac{a}{r}\right)$ inner graph copies by Property  2 of \cref{thm:modified_outer_graph}, it must suffer $+\Theta\left(\frac{a}{r}\right)$ additive error. 

    We can balance our parameters by letting $ \psi/c^2 = \Theta(\frac{a}{r})$. Since $ \psi/c^2 = \Theta_c\left( r^3 \right)$, this implies $a = \Theta_c\left( r^4 \right)$. Then $|V(G)| = \Theta_c(r^{17})$ by Property 1 of \cref{thm:spanner_obstacle_graph}, and the additive error suffered by spanner $H$ is at least $$\min\left(\Omega\left(\frac{\psi}{c^2}\right), \Omega\left(\frac{a}{r}\right)\right) = \Omega_c(r^3) = \Omega_c\left(|V(G)|^{3/17}\right).$$
\end{proof}

% \ghnoteinline{It seems unlikely that the subset additive spanners are going to make it by the ICALP deadline, unfortunately.}

% \begin{theorem}
%     For any sufficiently large parameter $c$, there are infinitely many $n$ for which there is an $n$-node graph $G$ and a set $S$ of $|S| = \Theta_c(n^{1/4})$ source nodes such that any pairwise spanner of $S \times S$ over $G$ with at most $cn$ edges has additive distortion $+c$.  
% \end{theorem}
% \begin{proof}

% \end{proof}

\bibliographystyle{alpha}
\bibliography{ref}

\appendix

\section{Proof of \cref{lem:sourcewise_dp} }
\label{app:sslb}

The goal of this section will be to prove the following lemma.

\begin{lemma}[cf. Theorem 5.10 of \cite{CoppersmithE06}] 
% \label{lem:sourcewise_dp}
For any $a, c > 0 \in \mathbb{Z}$, there exists a graph $G_I(a, c) = (V_I, E_I)$ with a set $S_I \subseteq V_I$ of sources, a set $T_I \subseteq V_I$ of sinks, and a set $P_I \subseteq S_I \times T_I$ of critical pairs that has the following properties:
\begin{enumerate}
\item The number of nodes, edges,  sources, sinks, and critical pairs in $G_I$ is:
\begin{align*}
    |V_I| &= \Theta(a^2), \\
    |E_I| &= \Theta(a^2c), \\
    |S_I| &= \Theta(a^{1/2}c^{11/4}), \\
    |T_I| &= \Theta(a^{1/2}c^{11/4}), \\
        |P_I| &= \Theta(ac^{5/2}).
\end{align*}
\item Every path $\pi_{s, t}$, where $(s, t) \in P_I$, contains $\Theta(a/c^{3/2})$ edges that do not lie on any other path $\pi_{s', t'}$, where $(s', t') \in P_I$. 
\item For every source $s \in S_I$ and sink $t \in T_I$, the distance between $s$ and $t$ in $G_I$ satisfies the following:
$$
 \dist_{G_I}(s, t) = \Theta(ac^{1/2}).
$$

\item The set of sources $S_I$ can be partitioned into $b = \Theta(c^3)$ sets $S_I^1, \dots, S_I^b$, where $|S_I^i| = \Theta(a^{1/2}c^{-1/4})$ for all $i \in [b]$. Let $T_I^i = \{t \in T_I \mid (S_I^i \times \{t\}) \cap P_I \neq \emptyset \}$ be the set of all sinks that belong to a critical pair with a source in $S_I^i$. Then for all $i \in [b]$  the following properties hold:
\begin{itemize}
    \item $|T_I^i| = \Theta(a^{1/2}c^{-1/4})$ for all $i \in [b]$,
    \item $S_I^i \times T_I^i \subseteq P_I$, and
    \item for all  $(s, t) \in P_I$ such that $s \in S_I^i$ and $t \in T_I^i$, 
$$
\dist_{G_I}(s, t) \leq \dist_{G_I}(S_I^i, T_I^i),$$
where $\dist_{G_I}(S_I^i, T_I^i)$ denotes the minimum distance between $S_I^i$ and $T_I^i$ in $G_I$. 
\end{itemize}
\end{enumerate}
\end{lemma}

% This lemma is largely implied by the  subset distance preserver lower bounds of \cite{CoppersmithE06}.
% For completeness, we present the construction of $G_I$ and the proof of its desired properties in \cref{subsec:innergraph_construction}. 

\subsection{Construction of $G_I$}
\label{subsec:innergraph_construction}
Let $a, c >0 \in \mathbb{Z}$ be the input parameters for our construction of inner graph $G_I = (V_I, E_I)$.  Parameter $a$ will grow with the input size, while $c$ will be assumed to be a sufficiently large integer much smaller than $a$.  For
simplicity of presentation, we will frequently ignore issues related to non-integrality of expressions
that arise in our analysis; these issues affect our bounds only by lower-order terms.
The nodes of $G_I$ will correspond to integer points in $\mathbb{R}^2$, while the edges in $G_I$ will correspond to vectors in a set $W$, where we define $W$ to be $$W = \{(w, w^2) \mid w \in [c/2, c] \}.$$
Let $\vec{w}_i = (w_i, w_i^2)$, where $w_i = c/2 - 1 + i$,  denote the $i$th vector of $W$ for $i \in [1, c/2]$.
\paragraph{Vertex Set $V_I$.}
Our vertex set $V_I$ will initially be defined to be the set of points
$$
V_I = \left[ ac^{1/2}\right] \times \left[\frac{a}{c^{1/2}}\right] \subseteq \mathbb{R}^2.
$$
We will add more nodes to $V_I$ in a later step of the construction.
\paragraph{Edge Set $E_I$.}
Our edge set $E_I$ in $G_I$ will be initially defined to be the set 
$$
E_I = \{(u, v) \in V_I \times V_I \mid v - u \in W\}.
$$
For analysis purposes, we will additionally add the edges $(v, v \pm (0, 1))$ and $(v, v \pm (1, 0))$ to every vertex $v \in V_I$. This will become relevant in the proof of \cref{claim:factors}. 
We will remove some of the edges in $E_I$ from $G_I$ in a later step of the construction.

\paragraph{Rectangles $R_i^S, R_i^T$ of $G_I$. }
In order to define our sources, sinks, and critical pairs, it will be helpful to define a collection of $c$ pairs of sets of nodes $(R_1^S, R_1^T), \dots, (R_{|W|}^S, R_{|W|}^T) \subseteq V_I \times V_I$. We refer to sets $R_i^S, R_i^T \subseteq V_I$, $i \in [|W|]$, as \textit{rectangles} of $G_I$. Each  rectangle of $G_I$ will correspond to the collection of points in $V_I$ contained in a `true' rectangle in $\mathbb{R}^2$. 
We construct rectangles $R_i^S, R_i^T$ in $G_I$ as follows:
\begin{itemize}
    \item Let $$\Vec{x}_i = (w_i, 2w_i),$$ and let $\hat{\Vec{x}}_i$ denote the unit vector in the direction $\Vec{x}_i$.
    Let $$\Vec{y}_i = \left(w_i, -\frac{1}{2w_i}\right)$$ be a vector perpendicular to $\Vec{x}_i$, and let $\hat{\vec{y}}_i$ denote the unit vector in the direction $\Vec{y}_i$. The rationale behind our choice of associating vectors $\Vec{x}_i$ and $\Vec{y}_i$ with $\Vec{w}_i$ will become clear in the proof of \cref{claim:innergraph_usp}.
    \item Let  $p_1 = (0, 0)$, $p_2 = \frac{ac^{1/2}}{100} \cdot \hat{\Vec{x}}_i$, $p_3 =  \frac{c}{100} \cdot \hat{\vec{y}}_i$, and $p_4 = p_2 + p_3$. Let $\mathcal{R}_i = \conv(p_1, p_2, p_3, p_4)$. Then $\mathcal{R}_i \subseteq \mathbb{R}^2$ is a rectangle in $\mathbb{R}^2$ with one vertex at the origin. 
    \item Let $d$ be the positive integer $\lceil a/ 100\rceil$.  Let $\mathcal{R}_i^S = \mathcal{R}_i + (d, d)$. It is straightforward to verify that $\mathcal{R}_i^S \subseteq \conv(V_I)$.
    Let $f$ be the positive integer $f = \left\lceil\frac{a}{100c^{3/2}}\right\rceil$. Let $\mathcal{R}_i^T = \mathcal{R}_i^S + f \cdot \vec{w}_i$. Again it is straightforward to verify that $\mathcal{R}_i^T \subseteq \conv(V_I)$. 
    \item We define rectangles $R_i^S, R_i^T$ in $G_I$ as
    $$
    R_i^S = V_I \cap \mathcal{R}_i^S, \text{\qquad and \qquad} R_i^T = V_I \cap \mathcal{R}_i^T.
    $$
    \item Note that $|R_i^S| = |R_i^T| = \Theta(ac^{3/2})$, by taking $a$ and $c$ to be sufficiently large. 
\end{itemize}
The following claim is immediate by the construction of rectangles $R_i^S$ and $R_i^T$. 

\begin{claim}
\label{claim:rectangle}
    For each node $v \in R_i^S$ and each vector $\Vec{w}_j \in W$, $v + \Vec{w}_j \not \in R_i^S$.
\end{claim}

\paragraph{Bands $B_i^S, B_i^T$ of $G_I$.}
As the next step in defining our sources, sinks, and critical pairs, it will be helpful to define a collection of $b$ pairs of sets of nodes $(B_1^S, B_1^T), \dots, (B_b^S, B_b^T) \subseteq V_I \times V_I$, where $b = \Theta(c^3)$. We refer to sets $B_i^S, B_i^T \subseteq V_I$ as \textit{bands} of $G_I$.  We construct bands in $G_I$ as follows:
\begin{itemize}
    \item Fix a pair of rectangles $R_i^S, R_i^T$ in $G_I$, $i \in [|W|]$. 
    \item   Let  $q_1 = (0, 0)$, $q_2 = \frac{c}{100} \cdot \hat{\Vec{x}}_i$, $q_3 =  \frac{c}{100} \cdot \hat{\vec{y}}_i$, and $q_4 = q_2 + q_3$, where $\hat{\Vec{x}}_i$ and $\hat{\vec{y}}_i$ are as defined in the previous paragraph. 
    \item Let $\mathcal{B}_i = \conv(q_1, q_2, q_3, q_4) + (d, d)$ be a rectangle  in $\mathbb{R}^2$. Then $\mathcal{B}_i \subseteq \mathcal{R}_i^S$. Let $B_i = V_I \cap \mathcal{B}_i$.
    \item Note that $|B_i| = \Theta(c^2)$, by taking $a$ and $c$ to be sufficiently large. Let $v_1, \dots, v_{|B_i|}$ denote the nodes in $B_i$. 
    \item For each node $v_j \in B_i$ we will define associated bands $B_j^S$ and $B_j^T$. Given a $v_j \in B_i$, we define $B_j^S$ and $B_j^T$ as
    $$
    B_j^S = \{u \in R_i^S \mid u = v_j + k \cdot \Vec{x}_i, k \in \mathbb{Z}\} \text{\qquad and \qquad} B_j^T = B_j^S + f \cdot \vec{w}_i.
    $$
    \item Observe that $|B_j^S| = |B_j^T| = \Theta\left(\frac{a}{c^{1/2}}\right)$. We will assume  that $|B_j^S| = |B_j^T| = z^2$ for some positive integer $z \in \mathbb{Z}$. This is without loss of generality because we can  delete a constant fraction of elements in $B_j^S$ and $B_j^T$ to ensure this property holds. 
    \item We can define bands in this way for all pairs of rectangles $R_i^S, R_i^T$ in $G_I$, $i \in [|W|]$. This will yield a collection of $ b = |W| \cdot \Theta(c^2) = \Theta(c^3)$ distinct pairs of bands $B_i^S, B_i^T$.  This completes our construction of the bands of $G_I$. We summarize the properties of our collection of bands in the following claim.
\end{itemize}
\begin{claim}
\label{claim:bands}
Our collection of bands $\{B_i^S, B_i^T\}_{i \in [b]}$ satisfies the following properties:
\begin{itemize}
    \item $b = \Theta(c^3)$, 
    \item $|B_i^S| = |B_i^T| = \Theta\left( \frac{a}{c^{1/2}}\right)$, and $|B_i^S| = |B_i^T| = z^2$ for some $z \in \mathbb{Z}$. 
    \item Band $B_i^S$ and band $B_i^T$ are of the form $\{v + k \cdot \Vec{x}_i \mid k \in [\ell]\}$ for some $v \in V_I$ and $\ell = \Theta\left( \frac{a}{c^{1/2}}\right)$.
    \item For all $j,k \in [|B_i|]$ such that $j \neq k$, we have that $B_j^S \cap B_k^S = \emptyset$, since for all $u \in R_i^S$, there is at most one node $v_j \in B_i$ such that $u \in \{v_j + \ell \cdot \Vec{x}_i \mid \ell \in \mathbb{Z} \}$. Likewise, $B_j^T \cap B_k^T = \emptyset$. 
\end{itemize}
\end{claim}
\begin{proof}
    The first three properties follow immediately from the above discussion; what remains is to verify the fourth property. Suppose for the sake of contradiction that there exists a node $u \in R_i^S$ and distinct nodes $v, v' \in B_i$ such that 
    $$
    u = v + \ell_1 \cdot \Vec{x}_i = v' + \ell_2 \cdot \Vec{x}_i,
    $$
    where $\ell_1, \ell_2 \in \mathbb{Z}$ and $\ell_1 \neq \ell_2$. Then $\|v - v'\| \geq \|\Vec{x}_i\| > q_2 - q_1$, contradicting our assumption that $v, v' \in B_i$. 
\end{proof}

\paragraph{Band Factors $F_{i, j}^S, F_{i, j}^T$ of $G_I$.}
As the next step in defining our sources, sinks, and critical pairs, it will be helpful to partition each band $B_i^S$ in $G_I$ into  sets $F_{i, 1}^S, \dots, F_{i, \nu}^S$,  where $\nu = |B_i^S|^{1/2}$, 
and partition each band $B_i^T$ in $G_I$ into sets $F_{i, 1}^T, \dots, F_{i, \nu}^T$.  We refer to sets $F_{i, j}^S$ and $F_{i, j}^T$ as \textit{band factors} of $G_I$. We construct band factors in $G_I$ as follows:
\begin{itemize}
    \item Fix a pair of bands $B_i^S, B_i^T$, where $i \in [b]$. 
    \item By \cref{claim:bands}, there exists nodes $v, v' \in V_I$ and integer $\ell = |B_i^S| = \Theta\left( \frac{a}{c^{1/2}}\right)$ such that
        $$
        B_i^S = \{v + k \cdot \Vec{x}_i \mid k \in [\ell]\} \text{\qquad and \qquad } B_i^T = \{v' + k \cdot \Vec{x}_i \mid k \in [\ell]\}.
        $$
        Moreover, $\ell = z^2$ for some $z \in \mathbb{Z}$. 
        \item For $j \in [1, \ell]$, let $s_j \in B_i^S$ denote node $s_j = v+j \cdot \Vec{u}$, and let $t_j \in B_i^T$ denote node $t_j = v'+j \cdot \Vec{u}$. 
        \item Let $\nu = |B_i^S|^{1/2}$. We define band factors $F_{i, j}^S, F_{i, j}^T$, $j \in [\nu]$, as follows:
        $$
F_{i, j}^S = \{s_k \in B_i^S \mid k \in [(j - 1) \cdot \nu + 1, \quad j \cdot \nu ] \}  \text{\qquad and \qquad} F_{i, j}^T =  \{t_k \in B_i^T \mid k \equiv j \bmod \nu \}.
$$
We summarize the key properties of  band factors $F_{i, j}^S, F_{i, j}^T$ in the following claim.
\end{itemize}

\begin{claim}
    \label{claim:factors}
    Our collection of band factors $\{F_{i, j}^S, F_{i, j}^T\}_{i \in [b], j \in [\nu]}$ satisfies the following properties:
    \begin{enumerate}
        \item $|F_{i, j}^S| = |F_{i, j}^T| = \nu = \Theta\left( \frac{a^{1/2}}{c^{1/4}} \right)$ for all   $i \in [b], j \in [\nu]$,
        \item $\{F_{i, j}^S\}_{j \in [\nu]}$ (respectively,  $\{F_{i, j}^T\}_{j \in [\nu]}$)  is a partition of $B_i^S$ (respectively, $B_i^T$) for all $i \in [b]$, 
        \item $|\{k \in [\nu] \mid s_k \in F_{i, j}^S \text{ and }  t_k \in F^T_{i, j'}\}| = 1$ for all $i \in [b]$ and $j, j' \in [\nu]$, and
        \item $\dist_{G_I}(s_{k}, s_{k'}) = O(|k - k'| \cdot c)$ and $\dist_{G_I}(t_k, t_{k'}) = O(|k - k'| \cdot c)$ for all $k, k' \in [\nu]$.
    \end{enumerate}
\end{claim}
\begin{proof}
    Properties 1 and 2 are immediate from the construction of of our collection of band factors. What remains is to prove Properties 3 and 4. 

    We begin with Property 3. 
    Note that $F_{i, j}^S$ contains only the nodes $s_k$, where $k$ is in an interval $[(j-1) \cdot \nu + 1, j \cdot \nu]$ of size $\nu$. Moreover, since $F_{i, j}^S$ only contains nodes $t_k$  where $k \equiv j \bmod \nu$. Every interval of $\nu$ consecutive integers contains exactly one integer equivalent to $j$ modulo $\nu$, so there is exactly one integer $k$ such that $s_k \in F_{i, j}^S$ and $t_k \in F_{i, j}^T$. 
    
    Property 4 follows from two observations:
    \begin{itemize}
        \item We can travel from any node $v \in F_{i, j}^S \cup F_{i, j}^T$ to nodes $v \pm (0, 1) \in V_I$ and nodes  $v \pm (1, 0) \in V_I$ using a single edge in $G_I$. 
        \item $\|s_k - s_{k'}\| = |k - k'| \cdot \|\Vec{x}_i\| \leq |k - k'| \cdot 3c$ and $\|t_k - t_{k'}\|= |k - k'| \cdot \|\Vec{x}_i\| \leq |k - k'| \cdot 3c$. 
    \end{itemize}
    Combining these observations, we conclude that for all $k, k' \in [\nu]$,  there exists an $s_k \leadsto s_{k'}$ path 
 and a  $t_k \leadsto t_{k'}$ path of  each of length at most $O(|k - k'| \cdot c)$ in $G_I$. 
\end{proof}

\paragraph{Trees $T_{i, j}^S, T_{i, j}^T$ of $G_I$.} As the final step in defining our sources, sinks, and critical pairs, we will construct a binary tree with subdivided edges for each band factor in $G_I$. Band factors $F_{i, j}^S$ and $F_{i, j}^T$ will have associated trees $T_{i, j}^S$ and $T_{i, j}^T$, respectively. The leaves of each tree associated with each band factor will be merged with the nodes of the band factor. We will first construct the trees $T_{i, j}^S, T_{i, j}^T$, and will later merge them with their respective band factors in a subsequent step. We define trees $T_{i, j}^S, T_{i, j}^T$ to be disjoint copies of the following tree:
\begin{itemize}
    \item Let $T$ be a binary tree with root $r$ such that
    \begin{itemize}
        \item $T$ has $\nu$ leaves and
        \item every leaf of $T$ has depth $d = \Theta(\log \nu)$.
    \end{itemize}
        % \item Replace each edge $e$ in $T$ with a path $P_e$ of length $\psi = \lceil a^{1/2}c \rceil$.
        \item For each edge $e \in E(T)$, if the parent node incident to $e$ has depth $i$, then replace $e$ with a path $P_e$ of length $\psi_i = 2^{d-i} \cdot a^{1/2}c$. 
        \item Let $\ell_1, \dots, \ell_{\nu}$ be the leaves of binary tree $T$ labeled from left to right. The key properties of tree $T$ are summarized in the following claim.
\end{itemize}
\begin{claim}
\label{claim:trees}
    Tree $T$ satisfies the following properties:
    \begin{enumerate}
        \item $T$ has $\nu$ leaves and $|T| = \Theta(  ac^{1/2}  \log a)$ nodes,
        \item every root to leaf path in $T$ is of length $\psi :=  \sum_{i=0}^d \psi_i = \Theta(ac^{1/2})$, and
        \item $\dist_T(\ell_k, \ell_{k'}) = \Omega( |k - k'| \cdot a^{1/2}c)$ for all $k, k' \in [\nu]$.
    \end{enumerate}
\end{claim}
\begin{proof} \leavevmode
\begin{enumerate}
    \item $T$ has $\nu$ leaves by construction. The number of nodes in the long edges of depth $i$ in $T$ is $2^i \cdot \psi_i = \Theta(\nu \cdot a^{1/2}c)$. Since the original binary tree had depth $\Theta(\log \nu)$, the claim follows.
    \item Every root to leaf path has length $\psi := \sum_{i=0}^d \psi_i = \Theta(a^{1/2}c \cdot \nu)$, as claimed.
    \item Note that the lowest common ancestor of leaves $\ell_k$ and $\ell_{k'}$ in the original binary tree had depth at most $d' =d - \lceil \log_2(|k - k'|)\rceil$. Then the distance in $T$ from $\ell_k$ to $\ell_{k'}$ is at least
    $$\dist_T(\ell_k, \ell_{k'}) \geq \sum_{i=d'}^d \psi_i = a^{1/2}c \cdot \sum_{i=d'}^d 2^{d-i} = a^{1/2}c \cdot \Omega(2^{ \log(|k - k'|)}) = \Omega(a^{1/2}c \cdot |k - k'|).$$
\end{enumerate}
\end{proof}
We define $T_{i, j}^S$ and $T_{i, j}^T$ to be copies of the tree $T$ constructed above for all $i \in [b]$ and $j \in [\nu]$. We let $r(i, j, S)$ denote the root of $T_{i, j}^S$, and we let $\ell_k(i, j, S)$ denote the $k$th leaf of $T_{i, j}^S$. We define identical notation for $T_{i, j}^T$. We will now merge our collection of trees $T_{i, j}$ with $G_I$, by fusing the leaves of the trees with their associated band factors. This will complete our construction of graph $G_I$. 
\begin{itemize}
    \item For all $i \in [b]$, $j \in [\nu]$, and $k$, merge leaf $\ell_k(i, j, S)$ in tree $T_{i, j}^S$ with the $k$th node 
 in band factor $F_{i, j}^S$. Likewise, merge  leaf $\ell_k(i, j, T)$ in tree $T_{i, j}^T$ with the $k$th node
 in band factor $F_{i, j}^T$.
 \item This completes our construction of graph $G_I$.
 We now verify two simple claims about  $G_I$. In particular, we will  verify that inserting our trees into $G_I$ did not decrease distances in $G_I$ between nodes in $V_I$. 
\end{itemize}
\begin{claim}
Distances in $G_I$ between nodes in $V_I$ do not decrease after inserting trees $\{T_{i, j}^S, T_{i, j}^T\}_{i \in [b], j \in [\nu]}$ into $G_I$.
\label{claim:tree_decrease}
\end{claim}
\begin{proof}
    Let $v, v'$ be the $k$th and $k'$th nodes in $F^S_{i, j}$. 
    By \cref{claim:factors} and the definition of $F^S_{i, j}$, the distance in $G_I$ from $v$ to $v'$ is at most $O(|k - k'| \cdot c)$. On the other hand, the distance in $T_{i, j}^S$ from $\ell_k(i, j, S)$ to $\ell_{k'}(i, j, S)$ is at least $\Omega(a^{1/2}c \cdot |k - k'|)$. This implies the inequality
    $$
    \dist_{T_{i, j}^S}( \ell_k(i, j, S), \ell_{k'}(i, j, S)) \geq \dist_{G_I}(v, v').
    $$
    Then inserting tree $T_{i, j}^S$ into $G_I$ will not decrease distances in $G_I$ between nodes in $V_I$. 
    
    Likewise, let $v, v'$ be the $k$th and $k'$th nodes in $F_{i, j}^T$. By \cref{claim:factors} and the definition of $F_{i, j}^T$, the distance in $G$ from $v$ to $v'$ is at most $O(|k - k'| \cdot c \cdot \nu) = O(|k-k'| \cdot a^{1/2} c^{3/4})$. 
    On the other hand, the distance in $T_{i, j}^S$ from $\ell_k(i, j, S)$ to $\ell_{k'}(i, j, S)$ is at least $\Omega(a^{1/2}c \cdot |k - k'|)$. Then by taking $a$ and $c$ to be sufficiently large, we can guarantee that  
    $$
    \dist_{T_{i, j}^S}( \ell_k(i, j, S), \ell_{k'}(i, j, S)) \geq \dist_{G_I}(v, v').
    $$
    We conclude that distances in $G_I$ between nodes in $V_I$ do not decrease after inserting  trees $\{T_{i, j}^S, T_{i, j}^T\}_{i \in [b], j \in [\nu]}$ into $G_I$.
\end{proof}

\begin{claim}
\label{claim:node_size}
    $|V(G_I)| = \Theta(a^2)$
\end{claim}
\begin{proof}
    Note that $|V_I| = a^2$ by construction. What remains is to bound the contribution of the nodes in our collection of trees $\{T_{i, j}^S, T_{i, j}^T\}$ in $G_I$. There are $\Theta(b\nu) = \Theta\left(c^3 \cdot \frac{a^{1/2}}{c^{1/4}} \right)$ trees in our collection of trees. Moreover, each tree has $\Theta(ac^{1/2} \log a)$ nodes by \cref{claim:trees}. Then our trees contribute $\Theta(a^{3/2}c^{13/4} \cdot \log a) = o(a^2)$ nodes to $G_I$.  
\end{proof}

\paragraph{Sources $S_I$, Sinks $T_I$, and Critical Pairs $P_I$.} We define our sources $S_I$ and sinks $T_I$ as follows:
\begin{itemize}
    \item let $S_I = \{r(i, j, S) \mid i \in [b], j \in [\nu] \}$,
    \item let $S_I^i = \{r(i, j, S) \mid j \in [\nu]\}$ for  all $i \in [b]$,
    \item let $T_I = \{r(i, j, T) \mid i \in [b], j \in [\nu] \}$,
    \item let $T_I^i = \{r(i, j, T) \mid j \in [\nu]\}$ for  all $i \in [b]$, and 
    \item let $P_I = \cup_{i \in [b]} S_I^i \times T_I^i$.
\end{itemize}
The following claim is immediate.
\begin{claim}
    \label{claim:source_sizes}
    Our sources, sinks, and critical pairs have the following sizes:
    \begin{itemize}
            \item $|S_I| = \Theta\left(a^{1/2}c^{11/4} \right)$, 
        \item $|S_I^i| = \Theta(a^{1/2}c^{-1/4})$ for all $i \in [b]$,
       \item $|T_I| = \Theta\left(a^{1/2}c^{11/4} \right)$, 
    \item $|T_I^i| = \Theta(a^{1/2}c^{-1/4})$ for all $i \in [b]$, and
    \item $|P_I| = \Theta(ac^{5/2})$. 
    \end{itemize}
\end{claim}

\paragraph{Critical Paths for $P_I$.}
For each index $i \in [b]$, and each node $s \in S_I^i$ and each $t \in T_I^i$, we will define a critical path $\pi_{s, t}$ in $G_I$ as follows.
\begin{itemize}
    \item Let $F_{i, j}^S$ and $T_{i, j}^S$ be the band factor and tree associated with node $s \in S_I^i$ (note in particular, $s$ is the root $r(i, j, S)$ of $T_{i, j}^S$). Likewise, let  $F_{i, j'}^T$ and $T_{i, j'}^T$ be the band factor and tree associated with node $t \in T_I^i$.
    \item By Property 3 of \cref{claim:factors}, there is a unique node $s_k \in F_{i, j}^S$ and $t_k \in F_{i, j'}^T$ such that $s_k$ is the $k$th node of $B_i^S$ and $t_k$ is the $k$th node of $B_i^T$. 
    \item By the construction of $B_i^S$ and $B_i^T$,  it follows that $t_k = s_k + f \cdot \vec{w}_i$. 
    \item Note that since $s_k \in F_{i, j}^S$ and $t_k \in F_{i, j'}^T$, node $s_k$ was merged with a leaf in $T_{i, j}^S$ and node $t_k$ was merged with a leaf in $T_{i, j}^T$. 
    \item We define our critical path $\pi_{s, t}$ as 
    $$
    \pi_{s, t} = T_{i, j}^S[s, s_k] \circ (s_k + \vec{w}_i, s_k + 2\vec{w}_i, \cdots, s_k + (f-1) \cdot  \vec{w}_i, s_k + f \cdot \vec{w}_i) \circ T_{i, j'}^T[t_k, t].
    $$
    \item Observe that for all $(s, t) \in S_I^i \times T_I^i \subseteq P_I$, we have that $|\pi_{s, t}| =(f-1) +  2\sum_{i=0}^d \psi_i = \Theta(f + ac^{1/2}) = \Theta(ac^{1/2})$, by \cref{claim:trees}. In particular, this implies that $|\pi_{s, t}| = |\pi_{s', t'}|$ for all $(s, t), (s', t') \in S_I^i \times T_I^i$. 
    \item Let $\Pi_I = \{\pi_{s, t} \mid (s, t) \in P_I\}$ denote the critical paths associated with critical pairs $P_I$.
\end{itemize}
As a final step in our construction of $G_I$, we will remove all edges in $G_I$ that do not lie on a critical path $\pi_{s, t} \in \Pi_I$.
We now finish our analysis by proving a few claims about  critical paths  $\pi_{s, t} \in \Pi_I$.

\begin{claim}
\label{claim:innergraph_usp}
Path $\pi_{s, t}$ is a unique shortest path for all $(s, t) \in P_I$. 
\end{claim}
\begin{proof}
% The proof of this claim will follow from  arguments similar to those in \cref{lem:strongly-convex} and \cref{claim:outer_graph_unique_path}. Specifically, by an identical argument as in \cref{lem:strongly-convex}, we have that
% $$
% \Vec{w}_i
% $$
Let $s \in S_I^i$ and let $t \in T_I^i$. Let $F_{i, j}^S$ and $T_{i, j}^S$ be the band factor and tree associated with node $s \in S_I^i$, and let $F_{i, j'}^T$ and $T_{i, j'}^T$ be the band factor and tree associated with node $t \in T_I^i$. Let $\Vec{w}_k  = (w_k, w_k^2) \in W$ be the vector associated with the construction of path $\pi_{s, t}$. We observe that by an identical argument to \cref{lem:strongly-convex}, vector $\Vec{w}_k \in W$ is the unique vector in $W \cup -W$ that maximizes the inner product of $\langle \Vec{w}_k, \Vec{x}_k \rangle$.

Let $\pi$ be an $s \leadsto t$ path in $G_I$. Let $\pi = \pi_1 \circ \pi_2 \circ \pi_3$, where $\pi_1 \subseteq T_{i, j}^S$ and $\pi_3 \subseteq T_{i, j}^T$. Then  $|\pi_1| = |\pi_3| = \sum_{i=0}^d \psi_i = \Theta(ac^{1/2})$ without loss of generality.
What remains is to argue that if $\pi \neq \pi_{s, t}$, then $|\pi_2| > \pi_{s, t}[V_I]$, where $\pi_{s, t}[V_I]$ denotes path $\pi_{s, t}$ induced on $V_I$.

Note that $|\pi_{s, t}[V_I]| = f-1$ by the earlier discussion in the construction of critical paths for $P_I$.  Let $x, y$ be the endpoints of $\pi_2$.  Then by applying an  argument similar to \cref{claim:outer_graph_unique_path}, if $\pi_2$ contains an edge $e = (u, v)$ where $v - u \neq \Vec{w}_k$, then $|\pi_2| > |\pi_{s, t}[V_I]|$. 
Then $\pi_2$ contains only edges $(u, v)$ where $v - u = \Vec{w}_k$. (Note that $\pi_{s, t}[V_I]$ also contains only edges $(u, v)$ where $v - u = \Vec{w}_k$.) Let $x'$ and $y'$ be the endpoints of path $\pi_{s, t}[V_I]$. If $x = x'$ and $y = y'$, then $\pi_{s, t}[V_I] = \pi_2$. Then $x \neq x'$ or $y \neq y'$.  However, recall that by Property 3 of \cref{claim:factors}, there is exactly one index $k \in [\nu]$ such that $s_k \in F_{i, j}^S$ and $t_k \in F_{i, j'}^T$. Then this implies that $x = x'$ and $y = y'$, a contradiction. We conclude that if $|\pi_2| = |\pi_{s, t}[V_I]|$, then $\pi_2 = \pi_{s, t}[V_I]$, so $\pi$ is a unique shortest $s \leadsto t$ path in $G_I$.  
\end{proof}

\begin{claim}
    Every critical path $\pi_{s, t}$, where $(s, t) \in P_I$, contains $\Theta(a/c^{3/2})$ edges that do not lie on any other path $\pi_{s', t'}$, where $(s', t') \in P_I$. 
    \label{claim:disjoint_paths}
\end{claim}
\begin{proof}
Every edge $e = (u, v)$ in $G_I[V_I]$ belongs to at least one path $\pi_{s, t}$ in $\Pi_I$ by the final step of our construction. By construction, $\pi_{s, t}$ contains a node $x$ in rectangle $R_i^S$ for some $i \in [|W|]$ such that $x = u - i \cdot (u - v)$ for some positive integer $i$. 
Then edge $e$ belongs to at most one path in $\pi_{s, t} \in \Pi_I$ because any additional path in $\Pi_I$ containing $e$ must also contain a node $x'$ in  $R_i^S$ such that $x' = u - i' \cdot (u - v)$ for some positive integer $i'$. However, this is not possible by our choice of $R_i^S$ due to \cref{claim:rectangle}.

We have shown that every edge $e \in G_I[V_I]$ belongs to exactly one path $\pi_{s, t}$, where $(s, t) \in P_I$. Now note that each path has $f-1 = \Theta(a/c^{3/2})$ paths in $G_I[V_I]$, so the claim follows.
\end{proof}

We are finally ready to wrap up our proof of \cref{lem:sourcewise_dp}.

\begin{proof}[Proof of \cref{lem:sourcewise_dp}]
    \leavevmode

    \begin{enumerate}
        \item The size bounds of $|V_I|, |S_I|, |T_I|,$ and $|P_I|$ follow from \cref{claim:node_size} and \cref{claim:source_sizes}. To bound the size of $E_I$, note that by \cref{claim:disjoint_paths}, each path $\pi_{s, t}$, where $(s, t) \in P_I$ contains $\Theta(a/c^{3/2})$ edges that do not lie on any other path $\pi_{s', t'}$, where $(s', t') \in P_I$. Then $|E_I| = \Theta(|P_I| \cdot \frac{a}{c^{3/2}}) = \Theta(a^2c)$, as claimed.
        \item This property follows from \cref{claim:disjoint_paths}.
        \item Note that for all $s \in S_I$ and $t \in T_I$, the distance between $s$ and $t$ in $G_I$ is at least $\Omega(ac^{1/2})$ by \cref{claim:trees}. Moreover, the diameter of $G_I[V_I] = \Theta(a)$ by the argument in  \cref{claim:factors}. Then $\dist_{G_I}(s, t) = \Theta(ac^{1/2})$. 
        \item This property follows from \cref{claim:source_sizes}, the construction of critical pairs $P_I$ and critical paths $\Pi_I$, and the argument in \cref{claim:innergraph_usp}. 
    \end{enumerate}
\end{proof}

\end{document}